\documentclass[11pt,a4paper]{article}
\usepackage{amsmath, amsfonts, amssymb}
\usepackage{a4wide}
\usepackage{parskip}
\usepackage{enumitem}
\usepackage{xcolor}

\usepackage{hyperref}
\usepackage{booktabs}% nicer horizontal lines
\usepackage{caption}% fix vertical spacing of table captions
\usepackage{siunitx}% align numbers at the decimal point
\usepackage{braket}
\usepackage{tabularx}
\usepackage{graphicx}
\usepackage{adjustbox}
\usepackage{multirow}
\usepackage{amsthm}
\usepackage{bm}
\usepackage{lscape}
\usepackage{float}
\usepackage[noadjust]{cite}

\def\qed{\hfill {$\square$}\goodbreak \medskip}

\theoremstyle{theorem}
\newtheorem{theorem}{Theorem}
\newtheorem{corollary}{Corollary}
\theoremstyle{definition}
\newtheorem{definition}{Definition}

\theoremstyle{lemma}
\newtheorem{lemma}{Lemma}
\newtheorem{proposition}[theorem]{Proposition}

\theoremstyle{remark}
\newtheorem{remark}{Remark}
\usepackage{tikz,xcolor,hyperref}

\definecolor{lime}{HTML}{A6CE39}
\DeclareRobustCommand{\orcidicon}{%
	\begin{tikzpicture}
		\draw[lime, fill=lime] (0,0)
		circle [radius=0.16]
		node[white] {{\fontfamily{qag}\selectfont \tiny ID}};
		\draw[white, fill=white] (-0.0625,0.095)
		circle [radius=0.007];
	\end{tikzpicture}
	\hspace{-2mm}
}

\foreach \x in {A, ..., Z}{%
	\expandafter\xdef\csname orcid\x\endcsname{\noexpand\href{https://orcid.org/\csname orcidauthor\x\endcsname}{\noexpand\orcidicon}}
}

\begin{document}
	\date{}
	{\vspace{0.01in}
		\title{Construction of codes over a commutative non-unital ring from simplicial complexes and their applications}
		\author{{\bf Vidya Sagar\footnote{email: {\tt vsagariitd@gmail.com }},\; \bf Shikha Patel\footnote{email: {\tt shikha\_1821ma05@iitp.ac.in}}\; and \bf Sanjay Kumar Singh\footnote{email: {\tt sanjayks@iiserb.ac.in}}} \\~\\
			$^{\ast \ddagger}$Department of Mathematics,\\ Indian Institute of Science Education and Research Bhopal\\
			Bhopal - 462066, Madhya Pradesh, India\\~\\
			$^\dagger$Department of Mathematics,\\ 
			Indian Institute of Information Technology Bhopal\\
			Bhopal - 462003, Madhya Pradesh, India}
		
		\maketitle
		\begin{abstract}
			In this article, we investigate the construction of linear codes over a finite ring $\mathcal{S}$, where $\mathcal{S}$ is taken to be an extension of a commutative non-unital ring $I$ of order $p^2$. Our approach is based on the defining set method. The defining sets considered in this work are derived from general simplicial complexes that may contain multiple maximal elements. We determine the parameters of these codes over $\mathcal{S}$ and study their Gray images. We also study the corresponding subfield-like codes. We show that these Gray image codes and subfield-like codes produce several families of divisible codes. Furthermore, we establish sufficient conditions under which these codes are minimal, optimal, and self-orthogonal. As applications of our results, we obtain several families of projective few-weight codes, and locally recoverable codes with small locality. We also study the minimal access structures of secret-sharing schemes associated with the duals of these minimal codes. Moreover, we construct several families of strongly regular graphs from projective two-weight codes and determine their parameters explicitly.

			\medskip
			
			
			\medskip
			
			
		\end{abstract}
		
		\section{Introduction}
        Linear codes over finite fields and finite rings form one of the most important classes of error-correcting codes due to their rich algebraic structure. Determining and constructing codes with optimal parameters is a central problem in coding theory, with significant implications for both theory and applications. %Linear codes over rings of size four such as  $\mathbb{Z}_4$ \cite{Z4}, $\mathbb{F}_4$ \cite{F4}, $\mathbb{F}_2 \times \mathbb{F}_2$ \cite{F2xF2}, and $\mathbb{F}_2 + u\mathbb{F}_2$ \cite{F2uF2} have been extensively studied in algebraic coding theory. These are the only commutative unital rings among the eleven rings of order four classified by B. Fine \cite{Fine}. The remaining rings are non-unital, among which five are commutative, namely $B, C, H, I,$ and $J$, while the remaining two, namely $E$ and $F$, are non-commutative.
        More recently, attention has shifted to codes over non-unital rings. Alahmadi et al. \cite{AlahmadiI1,AlahmadiI2, SagarTIT, SagarDCC} initiated the study of linear and quasi self-dual (QSD) codes, with classifications for small lengths. This work was extended by Kim and Roe \cite{KimI}. Further work on the non-commutative ring $E$ was carried out in \cite{AlahmadiE1,AlahmadiE2}, focusing on the construction and classification of quasi self-dual (QSD) and self-orthogonal codes. Subsequent studies have further advanced the theory of codes over non-unital rings (see \cite{DNAnonunital,MinjaNonunital,LCDnonUnital}).\par 
		In \cite{DingN}, the authors introduced the following generic construction of linear codes over $\mathbb{F}_q$, which gives a nice tool for constructing linear codes. Let 
		\[
		D=\{\{d_1 < d_2 < \dots < d_n\}\}\subseteq \mathbb{F}_{q^m}
		\]
		be an ordered multiset. Then, the linear code over $\mathbb{F}_q$ of length $|D|=n$ is defined as
		\begin{equation*}
			C_D=\left\{ c_D(v)=\big({\rm Tr}_{q}^{q^m}(vd_1), {\rm Tr}_{q}^{q^m}(vd_2), \dots, {\rm Tr}_{q}^{q^m}(vd_n)\big) : v\in \mathbb{F}_{q^m}\right\},
		\end{equation*}
		where ${\rm Tr}_{q}^{q^m}(\cdot)$ denotes the trace function from $\mathbb{F}_{q^m}$ onto $\mathbb{F}_q$ \cite{Lidl}. The code $C_D$ is referred to as the \textit{trace code} over $\mathbb{F}_q$ with defining set $D$. By identifying $\mathbb{F}_{q^m}$ with $\mathbb{F}_{q}^{m}$, it follows that the above code is equivalent to
		\begin{equation}\label{DingConst}
			C_D=\left\{ c_D(v)=\big(\langle v,d_1\rangle_q,\langle v,d_2\rangle_q,\dots,\langle v,d_n\rangle_q\big) : v\in \mathbb{F}_{q}^{m}\right\},
		\end{equation}
		where $\langle v, d_i\rangle_q$ denotes the Euclidean bilinear form of $v$ and $d_i$ over $\mathbb{F}_q$ (see \cite{Anuradha}). Subsequently, the construction given in Eq.~\eqref{DingConst} was generalized from finite fields to finite rings in an analogous manner.\par

		Consider the commutative and non-unital ring $I$ of order $p^2$ defined in \cite{Fine}. For a  positive integer $r$ and prime number $p$, set $q=p^r$. Let $\mathcal{S}$ be an extension of the ring $I$ of degree $r$. Throughout the paper, we consider the linear code
		\begin{equation}\label{OurConst}
			C_D = \{(v\cdot d)_{d\in D} : v\in \mathcal{S}^m\}
		\end{equation}
		over $\mathcal{S}$ of length $\vert D\vert$, and $D \subseteq \mathcal{S}^m$ is an ordered finite multiset. The study of $C_D$ becomes more interesting when the defining set $D$ is constructed via simplicial complexes. This approach has recently produced several interesting linear codes (see \cite{Hyun_Kim,Hyuan, Hyun_Lee, AMC1, MShi21, shi_nonchain, shi_x2, Shi_guan, shi_qian, shi_xuan, mixed2, wu_zhu, Zhu_Wei, Sagar_DM1, Sagar_CCDS, SagarDCC, Sagar2026, NKewat, SagarIWSDA, Sagar_AMC2, Yadav2025, Yadav2026}), and is expected to yield codes with good parameters for suitable choices of rings and defining sets.
		
		The weight distribution of a linear code plays a key role in determining its error-detecting and error-correcting capability \cite{wchuffman, Klove}, while few-weight and minimal codes have important applications in combinatorics and secret-sharing \cite{Few1, Few2, Ding_Ding}. In 2024, linear codes over the non-commutative non-unital ring $E$ of order $4$ were investigated via simplicial complexes in \cite{SagarTIT}. More recently, this work was extended in \cite{Wu2026}, where the authors considered the non-commutative non-unital ring $E$ of order $p^2$ and studied the codes over an extension ring $\mathcal{R}$ of~$E$. They studied the Lee weight distributions and investigated the corresponding Gray image codes and subfield-like codes.
		
		In this work, inspired by the above research, we construct linear codes over the ring $\mathcal{S}$ of order $p^{2r}$, where $\mathcal{S}$ is an $r$-degree extension of the commutative non-unital ring $I$ \cite{Fine}, using defining sets arising from general simplicial complexes. We determine their Lee weight distributions. We also investigate the corresponding  Gray image codes and subfield-like codes. In addition, we obtain several divisible codes over $\mathbb{F}_q$ from these Gray image codes and subfield-like codes, and determine their parameters. We provide sufficient conditions under which these $\mathbb{F}_q$-linear codes are minimal, optimal, and self-orthogonal. As applications of our results, we construct several classes of locally  recoverable codes (LRCs) with locality $2$ and $3$, as well as two infinite families of projective few-weight codes. We further examine the minimal access structures of secret-sharing schemes (SSSs) derived from the duals of these minimal codes. Besides, we investigate several classes of strongly regular Cayley graphs associated with projective two-weight codes and explicitly determine their parameters. The complement graphs of these Cayley graphs are also shown to be strongly regular, and their corresponding parameters are determined.

        %%Organization
		%The remainder of this manuscript is organized as follows. Section \ref{Sec:2} provides the necessary preliminaries. In Section \ref{Sec:3}, we investigate the construction of linear codes over $\mathcal{S}$, where defining sets are obtained via general simplicial complexes, and obtain $\mathcal{S}$-linear codes and their parameters. Section \ref{Sec:4} presents explicit Lee wight distributions of these $\mathcal{S}$-linear codes and obtains divisible codes from these Gray image codes. Construction of subfield-like codes from $\mathcal{S}$-linear codes are studied in Section \ref{Sec:5} and the corresponding divisible codes are determined together with their parameters. We also establish certain conditions under which these codes are minimal and optimal. Section \ref{Sec:6} provides self-orthogonal codes from the Gray image codes and the subfield-like codes. Finally, Section \ref{Sec:7} presents applications of the constructed $\mathbb{F}_q$-linear codes obtained in this work. In this section, we obtain LRCs with locality $2$ and $3$ as well as two infinite family of projective few-weight codes, and determine the minimal access structures of SSS derived from the duals of these minimal codes. We also obtain several SRGs associated with projective two-weight codes in this section. Section \ref{Sec:8} concludes the paper.

        The remainder of this paper is organized as follows. Section~\ref{Sec:2} provides the necessary preliminaries. In Section~\ref{Sec:3}, we study the construction of linear codes over $\mathcal{S}$ using defining sets arising from general simplicial complexes, and obtain the corresponding $\mathcal{S}$-linear codes together with their parameters. Section~\ref{Sec:4} determines the explicit Lee weight distributions of these codes and obtains divisible codes through their Gray images. In Section~\ref{Sec:5}, we study subfield-like codes derived from the $\mathcal{S}$-linear codes and obtain several families of divisible codes. We further establish conditions under which these codes are minimal and optimal. Section~\ref{Sec:6} presents self-orthogonal codes obtained from both the Gray image codes and the subfield-like codes. Applications of these obtained $\mathbb{F}_q$-linear codes are discussed in Section~\ref{Sec:7}. In particular, we construct locally repairable codes (LRCs) with locality $2$ and $3$, derive two infinite families of projective few-weight codes, determine the minimal access structures of secret-sharing schemes (SSSs) associated with the duals of these minimal codes, and obtain several strongly regular graphs (SRGs) arising from projective two-weight codes. Finally, Section~\ref{Sec:8} concludes the paper.
		
		\section{Basic definitions and notations}\label{Sec:2}
		This section presents the fundamental definitions and concepts needed for the remainder of the paper.\par 
		Throughout this paper, let $p$ be a prime number and let $q=p^r$ for some $r\in \mathbb{N}$. We denote the finite field of order $q$ by $\mathbb{F}_q$. Let $I$ be the commutative non-unital ring of order $p^2$ given by
		\begin{equation}
			I = \langle a, b ~\vert ~ pa = 0,\ pb = 0,\ a^2 = b,\ ab = 0\rangle.
		\end{equation}
		We adopt the notation $I$ from the work of Fine \cite{Fine}. \par
		Let $\mathcal{S}$ be an extension of the ring $I$ of degree $r$, that is, $[\mathcal{S}: I]=r$. Then each element of $\mathcal{S}$ can be expressed as $as + bt$, where $s, t \in \mathbb{F}_q$.\par
		Suppose $\Phi : \mathcal{S} \longrightarrow \mathbb{F}_q^2$ is the \textit{Gray map} given by 
		\begin{equation}\label{PhiEq}
			\Phi(as + bt) = (t, s+t).
		\end{equation}
		This extends to a map from $\mathcal{S}^m$ to $(\mathbb{F}_q^{m})^2$ component-wise for any $m \in \mathbb{N}$, where $\mathcal{S}^m = a\mathbb{F}_q^m + b\mathbb{F}_q^m$. Note that the map $\Phi$ is one-one and onto.
		\begin{definition}
			An $\mathcal{S}$-submodule $C$ of $\mathcal{S}^m$ is called a \textit{linear code} over $\mathcal{S}$ of length $m$.
		\end{definition}
		For $v, w \in \mathbb{F}_q^m$, the \textit{Hamming weight} of $v$ denoted by $wt_{H}( v )$ is the number of non-zero entries in $v$. The \textit{Hamming distance} between $v$ and $w$ is $d_H(v, w) = wt_{H}(v-w)$.\\
		Let $x=a\alpha + b\beta$, $y=a\alpha' + b\beta' \in \mathcal{S}^m$, where $\alpha, \beta, \alpha', \beta'\in \mathbb{F}_q^m$. Then the \textit{Lee weight} of $x$ is $wt_{Lee}(x) = wt_{H}(\Phi(x)) = wt_{H}(\beta) + wt_{H}(\alpha +\beta )$. The \textit{Lee distance} between $x$ and $y$ is $d_{Lee}(x, y) = wt_{Lee}(x-y)$. Note that the image of an $\mathcal{S}$-linear code under the above Gray map is an $\mathbb{F}_q$-linear code. Suppose $C$ is an $\mathcal{S}$-linear code of length $m$. Let $A_i$ be the cardinality of the set that contains all codewords of $C$ having Lee weight $i$, $0\leq i\leq 2m$. Then the homogeneous polynomial in two variables \[Lee_C(X, Y) = \sum_{c\in C}X^{2m-wt_{Lee}(c)}Y^{wt_{Lee}(c)}\] is called the \textit{Lee weight enumerator} of $C$ and the string $(1, A_1,\dots ,A_{2m})$ is called the \textit{Lee weight distribution} of $C$. Similarly, we may define Hamming weight enumerator and Hamming weight distribution of a linear code over a finite field. Two $\mathbb{F}_q$-linear codes $C_1$ and $C_2$ are said to be \textit{equivalent} if
		$$ C_2=\{(a_1c_{\pi(1)}, a_2c_{\pi(2)}, \dots, a_nc_{\pi(n)}) \ \vert \ (c_1, c_2, \dots, c_n)\in C_1\},$$
		where $\pi$ is a permutation of $\{1, 2, \dots, n\}$ and each $a_i \in \mathbb{F}_q^{\ast}$. In addition, if the total number of $i\geq 1$ such that $A_i\neq 0$ is $l$, then $C$ is called an $l$-\textit{weight linear code}. Every $1$-weight linear code is an equidistant linear code. Bonisoli characterized all equidistant linear codes over finite fields in \cite{Bonisoli}.
		\begin{lemma}[\cite{Bonisoli}\textnormal{(Bonisoli)}]\label{Bonisoli}
			Suppose $C$ is a equidistant linear code over $\mathbb{F}_q$. Then $C$ is equivalent to the $r$-fold replication of a simplex code, possibly with added $0$-coordinates.
		\end{lemma}
		An $[n, k, d]$-linear code $C$ is called \textit{distance optimal} if there exist no $[n,k,d+1]$-linear codes (see \cite{wchuffman}). Next we recall the Griesmer bound.
		\begin{lemma}[\cite{Griesmer}
			\textnormal{(Griesmer Bound)}]\label{griesmerbound} 
			If $C$ is an $[n,k,d]$-linear code over $\mathbb{F}_q$, then we have
			\begin{equation}\label{GriesmerBound}
				\sum\limits_{i=0}^{k-1}\left\lceil \frac{d}{q^i}\right\rceil \leq n,
			\end{equation}  	
			where $\lceil \cdot \rceil$ denotes the ceiling function.
		\end{lemma}
		A linear code is called a \textit{Griesmer code} if equality holds in Eq. \eqref{GriesmerBound}. Note that every Griesmer code is distance optimal, but not conversely.
        
    The following result characterizes the minimum Hamming distance of a linear code via its parity-check matrix.
    \begin{theorem}[{\cite[Corollary~1.4.14]{wchuffman}}]\label{ThmDual}
    Let $C$ be a linear code over $\mathbb{F}_q$ with parity-check matrix $H$. Then the following hold:
    \begin{enumerate}
    \item The minimum Hamming distance of $C$ satisfies $d(C)\geq d$ if and only if every set of $d-1$ columns of $H$ is linearly independent over $\mathbb{F}_q$.

    \item The minimum Hamming distance of $C$ satisfies $d(C)\leq d$ if and only if there exists a collection of $d$ columns of $H$ that is linearly dependent.
    
    \item The minimum Hamming distance of $C$ is equal to $d$ if and only if every set of $d-1$ columns of $H$ is linearly independent over $\mathbb{F}_q$, while there exists a collection of $d$ columns of $H$ that is linearly dependent.
    \end{enumerate}
    \end{theorem}

    \begin{definition}
            Let $C$ be a linear code over $\mathbb{F}_q$. If the Hamming weight of every codeword in $C$ is divisible by a positive integer $l$, then $C$ is called \emph{$l$-divisible}. The largest such integer $l$ is referred to as the \emph{divisibility constant} of $C$.
    \end{definition}
		For $m\in \mathbb{N}$, we shall write $[m]$ to denote the set $\{1, 2,\dots ,m\}$. For $w\in \mathbb{F}_{q}^m$, the set $\textnormal{Supp}(w)=\{ i\in [m]: w_i \neq  0\}$ is called the \textit{support} of $w$. Note that the Hamming weight of $w\in \mathbb{F}_q^m$ is $wt_H(w)=\vert \textnormal{Supp}(w)\vert $. For $v,w\in \mathbb{F}_{q}^m$, one says that $v$ \textit{covers} $w$ if $\textnormal{Supp}(w)\subseteq \textnormal{Supp}(v)$. If $v$ covers $w$, we write $w\preceq v$.
		\begin{definition}
			Let $C$ be a linear code over $\mathbb{F}_q$. An element $v\in C\setminus \{0\}$ is called \textit{minimal} if $w\preceq v$ and $w\in C\setminus \{0\}$ $\implies$ $w = \lambda v$ for some $\lambda \in \mathbb{F}_q^{\ast}$. If each nonzero codeword of $C$ is minimal then $C$ is called a \textit{minimal code}.
		\end{definition}
		Now we recall a result from \cite{ABcondition} which is a sufficient condition for a linear code over $\mathbb{F}_q$ to be minimal.
		\begin{lemma}[\cite{ABcondition}\textnormal{(Ashikhmin-Barg)}]\label{minimal_lemma}
			Let $C$ be a linear code over $\mathbb{F}_q$. Denote by $wt_0$ and $wt_\infty$ the minimum and maximum Hamming weights, respectively, among all nonzero codewords of $C$. If $\frac{wt_0}{wt_{\infty}}> \frac{q-1}{q}$, then $C$ is minimal.
		\end{lemma}
		Here we give the standard orthogonality relation of additive characters over the vector space~$\mathbb{F}_p^r$.
		\begin{lemma}[\cite{Lidl}]\label{OrthChar}
			Let $\omega_p$ be a primitive $p$-th root of unity and let $x\in \mathbb{F}_p^r$. Then
			\begin{equation*}
				\sum_{u\in \mathbb{F}_p^r}\omega_p^{\langle u,x\rangle_p}
				=
				\begin{cases}
					p^r, & \text{if } x=0,\\[4pt]
					0, & \text{if } x\neq 0,
				\end{cases}
			\end{equation*}
			where $\langle u,x\rangle_p$ denotes the Euclidean inner product over $\mathbb{F}_p$.
		\end{lemma}
		Let $V$ be an $k$-dimensional $\mathbb{F}_q$-subspace of $\mathbb{F}_q^m$. The dual of $V$ is
		\[
		V^{\perp}=\{x\in \mathbb{F}_q^m:\langle x,y\rangle_q=0 \text{ for any } y\in V\}.
		\]
		
		It is clear that $V^{\perp}$ is an $(m-k)$-dimensional $\mathbb{F}_q$-subspace of $\mathbb{F}_q^m$. 
		Now we have the following lemma.
		\begin{lemma}
			Let $m, r$ be positive integers and $q=p^r$. Suppose $V$ is an $k$-dimensional $\mathbb{F}_q$-subspace of
			$\mathbb{F}_q^m$ and $u \in \mathbb{F}_p^r \setminus \{\textbf{0}\}$. Then we have
			\[
			\sum_{x\in V} w_p^{\langle u,\langle y,x\rangle_q\rangle_p}
			=
			\begin{cases}
				q^k, & \text{if } y\in V^\perp;\\
				0, & \text{otherwise.}
			\end{cases}
			\]
		\end{lemma}

		\subsection{Simplicial Complexes}
		In this subsection, we begin by recalling the basic definition and fundamental concepts of simplicial complexes that will be needed in the subsequent sections.
		\begin{definition}
			A subset $\Delta$ of $\mathbb{F}_{q}^m$ is called a \textit{simplicial complex} if $v\in \Delta, w\in \mathbb{F}_q^m$ and $w\preceq  v$ $\implies$ $w\in \Delta$. An element $v\in \Delta$ is called a \textit{maximal element of} $\Delta$ if $v$  is not properly covered
			by any other element of $\Delta$.
		\end{definition} 
		Consider the map $\psi: \mathbb{F}_q^m\longrightarrow 2^{[m]}$ is defined as $\psi(w)=\textnormal{Supp}(w)$, where $2^{[m]}$ denotes the power set of $[m]$. Observe that the map $\psi$ is onto but not one-one for $q\geq 3$, and ${\rm Supp}(w) = {\rm Supp}(\lambda w)$ for $\lambda \in \mathbb{F}_q^{\ast}$, $w \in \mathbb{F}_q^m$. Now onwards, we will write $w$ instead of Supp($w$) whenever $w$ is a maximal element of $\Delta$.
		
		A simplicial complex can have more than one maximal elements. Let $M\subseteq [m]$. The simplicial complex generated by $M$ is denoted by $\Delta_{M}$ and is defined as 
		\begin{equation}
			\Delta_{M}=\{w\in \mathbb{F}_{q}^m \vert \textnormal{ Supp}(w)\subseteq M\}.
		\end{equation}
		Note that $\psi^{-1}(M)$ is the only maximal element of $\Delta_M$, and $\vert \Delta_{M}\vert = q^{\vert M\vert}$. Observe that $\Delta_{M}$ is a vector space over $\mathbb{F}_q$ of dimension $\vert M \vert$.\par
		
		Suppose that $\Delta_M$ is a simplicial complex of $\mathbb{F}_q^m$ with one maximal element $M\subseteq [m]$, then one can immediately get that
		\[
		\Delta_M^{\perp}=\Delta_{[m]\setminus M},
		\]
		since $\Delta_M$ is an $|M|$-dimensional $\mathbb{F}_q$-subspace of $\mathbb{F}_q^m$.\par 
		
		We now recall a few lemmas from \cite{Wu2026}, which will be useful in the next section.
		\begin{lemma}[\cite{Wu2026}]
			Let $m, s$ be positive integers and $M_i\subseteq [m]$ for $i\in [s]$.
			\begin{enumerate}
				\item For any $i, j \in [s]$, if $\Delta_{M_i} \subseteq \Delta_{M_j}$, then $\Delta_{M_j}^{\perp}\subseteq \Delta_{M_i}^{\perp}$.
				
				\item
				\(
				\bigcap_{i\in [s]}\Delta_{M_i}^{\perp}
				=
				\Delta_{\bigcup_{i\in [s]}M_i}^{\perp}.
				\)
			\end{enumerate}
		\end{lemma}
		
		The following lemma can be proved using the principle of inclusion-exclusion.
		\begin{lemma}[\cite{Wu2026}]\label{generatinglemma}
			Suppose $\Delta\subseteq \mathbb{F}_{q}^m$ is a simplicial complex and $\mathcal{M}$ is the collection of its maximal elements. Then 
			\begin{equation}
				\vert \Delta\vert =\sum\limits_{\emptyset\neq S\subseteq \mathcal{M}}(-1)^{\vert S\vert +1}q^{\vert \cap S\vert},
			\end{equation}
			where $\cap S=\bigcap\limits_{M\in S}\textnormal{Supp}(M)$.
		\end{lemma}
		For $M\subseteq [m]$, define a $\{0,1\}$-valued function
		\(
		\Psi(\cdot \mid \Delta_M): \mathbb{F}_q^m \longrightarrow \{0,1\}
		\)
		by
		\begin{equation}\label{EqPsi}
			\Psi(\beta\mid \Delta_M)=
			\begin{cases}
				1, & \text{if } \beta\in \Delta_M^{\perp},\\
				0, & \text{otherwise},
			\end{cases}
		\end{equation}
		and let $\delta_{i,j}$ denote the Kronecker delta function:
		\begin{equation}\label{KronekEq}
			\delta_{i,j}=
			\begin{cases}
				1, & \text{if } i=j,\\
				0, & \text{otherwise}.
			\end{cases}
		\end{equation}
		Then we have the following results.
		\begin{lemma}[\cite{Wu2026}]\label{Lem-8}
			Let $q = p^r$ and $\Delta$ be a simplicial complex of  $\mathbb{F}_q^m$ with $\mathcal{M} = \{M_1, M_2, \ldots, M_s\}$ the set of maximal  elements of $\Delta$. For any $\beta \in \mathbb{F}_q^m$, we have
			
			\[
			\sum_{u\in\mathbb{F}_p^r}\ \sum_{t\in\Delta}
			w_p^{\langle u, \langle \beta, t\rangle_q\rangle_p}
			=
			\sum_{\emptyset \neq S \subseteq \mathcal{M}}
			(-1)^{\vert S \vert + 1}
			q^{\vert \cap S \vert }
			\left(
			1 + (q-1) \Psi(\beta \mid \Delta_{\cap S})
			\right)
			\]
			
			and
			
			\[
			\sum_{u \in \mathbb{F}_p^r}\ \sum_{t \in \Delta^c}
			w_p^{\langle u, \langle \beta, t \rangle_q\rangle_p}
			=
			q^m+(q-1)q^m\delta_{\mathbf{0}, \beta}
			-
			\sum_{u\in\mathbb{F}_p^r}\ \sum_{t \in \Delta}
			w_p^{\langle u, \langle \beta, t\rangle_q\rangle_p},
			\]
			where $\cap S = \bigcap_{M \in S} M,$ and $\Psi(\beta\mid \Delta_M)$ and $\delta_{i, j}$ are defined by Eqs. \eqref{EqPsi} and \eqref{KronekEq}, respectively.
		\end{lemma}
		\begin{lemma}[\cite{Wu2026}]\label{Lem-A}
			Let $m, s, r$ be positive integers, $q=p^r$ and $\Delta$ be a simplicial complex of $\mathbb{F}_q^m$ with $\mathcal{M} = \{M_1, M_2, \dots, M_s\}$ the set of maximal elements of $\Delta$. Define
			\[
			\mathcal{A}_{\Delta,\beta}
			=
			|\Delta|-\frac{1}{q}
			\sum_{u\in \mathbb{F}_p^r}
			\sum_{t \in \Delta}
			w_p^{\langle u, \langle \beta, t\rangle_q \rangle_p}
			\]
			and
			\(
			f_{\beta}(S)
			=
			(-1)^{\vert S \vert + 1}
			q^{\vert \cap S\vert - 1}
			(q-1)
			\big(1-\Psi(\beta|\Delta_{\cap S})\big),
			\)
			where $\beta\in \mathbb{F}_q^m$, $\cap S = \bigcap_{M \in S} M$, $S\subseteq \mathcal{M}$ and $\Psi(\beta|\Delta_M)$ is given by Eq. \eqref{EqPsi}. Then
			\(
			\mathcal{A}_{\Delta,\beta}
			=
			\sum_{\emptyset\neq S \subseteq \mathcal{M}}
			f_{\beta}(S),
			\)
			and
			\begin{enumerate}
				\item \label{Lem-A1}
				$\mathcal{A}_{\Delta,\beta}=0$ if and only if
				\(
				\beta\in \bigcap_{i\in [s]}\Delta_{M_i}^{\perp},
				\)
				that is,
				\(
				\beta\in \Delta_{\bigcup_{i\in [s]}M_i}^{\perp}.
				\)
				
				\item \label{Lem-A2}
				If $\beta\notin \Delta_{M_i}^{\perp}$ for some $i\in [s]$, then
				\(
				\mathcal{A}_{\Delta, \beta}\geq (q-1)q^{\vert M_i \vert - 1}.
				\)
				Moreover, the equality holds if
				\(
				\beta\notin \Delta_{M_i}^{\perp}
				\quad \text{and} \quad
				\beta\in \Delta_{\bigcup_{j \in [s]\setminus \{i\}}M_j}^{\perp}.
				\)
				
				\item \label{Lem-A3}
				\(
				\mathcal{A}_{\Delta, \beta}
				\leq
				(q-1)\sum_{i \in [s]}q^{\vert M_i \vert - 1},
				\)
				and the equality holds if and only if
				\(
				\beta\notin \bigcup_{i \in [s]}\Delta_{M_i}^{\perp}
				\)
				and
				\(
				\beta\in \bigcap_{1 \leq i<j \leq s}
				\Delta_{M_i\cap M_j}^{\perp}.
				\)
			\end{enumerate}
		\end{lemma}
		The following lemma presents the necessary and sufficient conditions on simplicial complexes under which some $\beta \in \mathbb{F}_q^m$ satisfies Lemma \ref{Lem-A}.
		\begin{lemma}[\cite{Wu2026}]\label{Lem-A+}
			Let $m, s$ be positive integers. Let $M_i, N_j \subseteq [m]$ and $\Delta_{M_i}, \Delta_{N_j}$ be the simplicial complexes of $\mathbb{F}_q^m$ for $i \in [s]$. Then we have 
			\begin{enumerate}
				\item \label{Lem-A+1}
				There exists $\beta \in \mathbb{F}_q^m$ such that $\beta \notin \Delta_{M_i}^{\perp}$ and $\beta \in \Delta_{\bigcup_{j\in [s]\setminus\{i\}} M_j}^{\perp}$ iff
				\(
				M_i \setminus \bigcup_{j\in [s]\setminus\{i\}} M_j \neq \emptyset;
				\)
				
				\item \label{Lem-A+2}
				There exists $\beta \in \mathbb{F}_q^m$ such that
				\(
				\beta \notin \bigcup_{i\in [s]} \Delta_{M_i}^{\perp}
				\)
				and
				\(
				\beta \in \bigcap_{1 \leq i < j\leq s}
				\Delta_{M_i \cap M_j}^{\perp}
				\)
				if and only if
				\(
				M_i \setminus \bigcup_{j\in [s] \setminus \{i\}} M_j \neq \emptyset
				\quad \text{for any } i\in [s].
				\)
			\end{enumerate}
		\end{lemma}

		%%%%%Section 3
		\section{Linear codes over the ring $\mathcal{S}$ by using simplicial complexes}\label{Sec:3}
		Consider a subset $D = \{\{d_1<d_2<\dots <d_n\}\} \subseteq \mathcal{S}^m$, an ordered multiset.
		Define
		\begin{equation}\label{OurConst}
			C_D = \{c_D(v)=\big(v\cdot d\big)_{d\in D} ~\vert~ v \in \mathcal{S}^m\},
		\end{equation}
		where $x\cdot y$ is defined by Eq.~\eqref{EqEIP}.\\
		Observe that $C_D$ is a linear code over $\mathcal{S}$ of length $\vert D \vert $. The ordered set $D$ is called the \textit{defining set} of $C_{D}$. Note that by changing the ordering of $D$, we will get a code that is permutation equivalent (see \cite{wchuffman}) to $C_D$.\\
		Now, we have $c_{D}: \mathcal{S}^m \longrightarrow C_{D}\subseteq \mathcal{S}^{\vert D \vert}$ defined by
		\begin{equation}\label{c_DMap}
			c_{D}(v)=\big(v\cdot d\big)_{d\in D}
		\end{equation}
		is a surjective homomorphism of modules over $\mathcal{S}$.\\	   
		Set $D_i\subseteq \mathbb{F}_q^m, i=1, 2$ and let $D = aD_1 + bD_2\subseteq \mathcal{S}^m$. Assume that $v = a\beta_1 + b\beta_2 \in \mathcal{S}^m$ and $x = at_1 + bt_2\in D = aD_1 + bD_2$, where $ \beta_1, \beta_2 \in \mathbb{F}_q^m$ and $t_i\in D_i, i = 1, 2$. We have
		\begin{equation}\label{EqEIP}
			\begin{split}
				v\cdot x & = \big(a\beta_1 + b\beta_2 \big) \cdot \big(at_1 + bt_2 \big)\\
				& = a^2\langle \beta_1, t_1\rangle_q + ab\langle \beta_1, t_2\rangle_q + ba\langle \beta_2, t_1\rangle_q + b^2\langle \beta_2, t_2\rangle_q\\
				& = b\langle \beta_1, t_1\rangle_q.
			\end{split}
		\end{equation}
		
		Then the Lee weight of $c_{D}(v)$ is
		\begin{equation}\label{Eq-Lee}
			\begin{split}
				wt_{Lee}(c_{D}(v)) = 
				& ~wt_{Lee}\big(\big(\big(a\beta_1 + b\beta_2 \big) \cdot \big(at_1 + bt_2 \big)\big)_{t_i\in D_i}\big)\\
				= & ~wt_{Lee}\big(\big(b\langle \beta_1 , t_1\rangle_q \big)_{t_i\in D_i}\big)\\
				= & ~wt_{H}\big(\Phi\big((b \langle \beta_1 , t_1\rangle_q )_{t_i\in D_i}\big)\big)\\
				= & ~wt_{H}\big((\langle \beta_1 , t_1\rangle_q )_{t_i\in D_i}\big) + wt_{H}\big(( \langle \beta_1 , t_1\rangle_q )_{t_i\in D_i}\big)\\
				= & ~2wt_{H}\big(( \langle \beta_1 , t_1\rangle_q )_{t_i\in D_i}\big).
			\end{split}
		\end{equation}
		We know that if $x \in \mathbb{F}_{q}^n$, then $wt_H(x)=0 \iff x=\textbf{0}\in \mathbb{F}_{q}^n $. Hence from Lemma \ref{OrthChar}, we have
		\begin{equation}\label{keyeq1}
			\begin{split}
				wt_{Lee}(c_{D}(v))
				& =  2\Big[\vert D \vert -\frac{1}{q} \sum\limits_{u \in \mathbb{F}^r_p}\sum\limits_{t_1 \in D_1}\sum\limits_{t_2\in D_2} w_p^{\langle u, \ \langle \beta_1, t_1\rangle_q\rangle_p}\Big]\\
				& =  2\Big[\vert D_1 \vert \vert D_2 \vert  - \frac{1}{q} \vert D_2\vert  \sum\limits_{u \in \mathbb{F}^r_p}\sum\limits_{t_1 \in D_1}  w_p^{\langle u,\ \langle \beta_1, t_1\rangle_q\rangle_p}\Big]\\
				& =  2\vert D_2\vert \Big[\vert D_1 \vert   - \frac{1}{q}   \sum\limits_{u \in \mathbb{F}^r_p}\sum\limits_{t_1 \in D_1}  w_p^{\langle u,\ \langle \beta_1, t_1\rangle_q\rangle_p}\Big].
			\end{split}
		\end{equation}
		\begin{remark}
			Note that in Eq.~\eqref{keyeq1}, for $\beta_1,t_1\in \mathbb{F}_q^m$ and $u\in \mathbb{F}_p^r$, we have $\langle \beta_1,t_1\rangle_q \in \mathbb{F}_q$. By identifying $\mathbb{F}_q=\mathbb{F}_{p^r}$ with $\mathbb{F}_p^r$, we will consider $\langle \beta_1,t_1\rangle_q$ as an element of $\mathbb{F}_p^r$, and hence the inner product
			$\langle u,\langle \beta_1,t_1\rangle_q\rangle_p$ is computed over $\mathbb{F}_p$.
		\end{remark}
		
		The following result describes the lengths, sizes, and minimum Lee weights of $C_D$ for five distinct choices of $D$, where the corresponding $C_D$ are $\mathcal{S}$-linear codes defined by Eq.~\eqref{OurConst}.
		
		%%%%%%%%2

		\begin{theorem}\label{Thm-1}
			Let $m, r, s, k \in \mathbb{N}$ and $q=p^r$. Suppose that $\Delta_1$ and $\Delta_2$ are simplicial complexes of $\mathbb{F}_q^m$ with
			\[
			\mathcal{M}=\{M_1,M_2,\dots,M_s\}
			\quad \text{and} \quad
			\mathcal{N}=\{N_1,N_2,\dots,N_k\}
			\]
			being the sets of maximal elements of $\Delta_1$ and $\Delta_2$, respectively. Let $M_i \setminus \cup_{j\in [s]\setminus \{i\}}M_j\neq \emptyset$ for any $i\in [s]$. 
			\begin{enumerate}
				\item \label{part:1}
				Let $D = a\Delta_1 + b\Delta_2 \subseteq \mathcal{S}^m$. Then $C_D$ is an $\mathcal{S}$-linear code of length $\vert D\vert =\vert \Delta_1 \vert \vert \Delta_2 \vert $ and size $q^{\vert \cup_{i\in [s]} M_i \vert}$ and minimum distance $d_L = 2\vert \Delta_2\vert (q-1) q^{{\rm min}_{i\in [s]} \{\vert M_i \vert\} -1}$.
				
				\item \label{part:2}
				Let $D = a\Delta_1 + b\Delta^c_2 \subseteq \mathcal{S}^m$. Then $C_D$ is an $\mathcal{S}$-linear code of length $\vert D\vert =\vert \Delta_1 \vert (q^m- \vert \Delta_2 \vert )$ and size  $q^{\vert \cup_{i\in [s]} M_i \vert}$and minimum distance $d_L = 2(q^m-\vert \Delta_2 \vert )(q-1)q^{{\rm min}_{i\in [s]} \{\vert M_i \vert\} -1}$.
				
				\item \label{part:3}
				Let $D = a\Delta^c_1 + b\Delta_2 \subseteq \mathcal{S}^m$. Then $C_D$ is an $\mathcal{S}$-linear code of length $\vert D\vert =(q^m - \vert \Delta_1 \vert ) \vert \Delta_2 \vert $ and size $q^{m}$ and minimum distance $d_L = 2\vert \Delta_2 \vert (q-1)\big(q^{m-1} - \sum_{ i \in [s]}q^{ \vert M_i \vert -1}\big)$.
				
				\item \label{part:4}
				Let $D = a\Delta^c_1 + b\Delta^c_2 \subseteq \mathcal{S}^m$. Then $C_D$ is an $\mathcal{S}$-linear code of length $\vert D\vert = (q^m- \vert \Delta_1 \vert ) (q^m - \vert \Delta_2 \vert )$ and size $q^{m}$ and minimum distance $d_L = 2(q^m - \vert \Delta_2 \vert )(q-1)\big(q^{m-1} - \sum_{ i \in [s]} q^{\vert M_i \vert -1}\big)$.
				
				\item \label{part:5}
				Let $D = a\Delta_1 + b\Delta_2 \subseteq \mathcal{S}^m$. Then $C_{D^c}$ is an $\mathcal{S}$-linear code of length $\vert D^c\vert =q^{2m}-\vert \Delta_1 \vert \vert \Delta_2 \vert $ and size $q^{m}$ and minimum distance $d_L = 2(q-1)\big(q^{2m-1} - \vert \Delta_2 \vert \sum_{ i \in [s]} q^{\vert M_i \vert -1}\big)$.
			\end{enumerate}
		\end{theorem}
		
		\begin{proof}
			We only provide the proof of part \ref{part:5} since the
			others can be proved by using a similar method. Note that $D^c = (a\Delta_1 + b\Delta_2 )^c = (a\Delta_1^c + b \mathbb{F}_q^m) \bigsqcup (a\Delta_1 + b\Delta^c_2)$, where $\bigsqcup $ denotes the disjoint union. By Eq.~\eqref{keyeq1} and Lemmas~\ref{Lem-8} and~\ref{Lem-A}, we obtain
			\begin{equation}\label{Eq.13}
				\begin{split}
					wt_{Lee}(c_{D}(v))
					& =  2\vert D_2\vert \Big[\vert D_1 \vert   - \frac{1}{q}   \sum\limits_{u \in \mathbb{F}^r_p}\sum\limits_{t_1 \in D_1}  w_p^{\langle u, \langle \beta_1, t_1\rangle_q\rangle_p}\Big]\\
					& = 2 q^m \Big[\vert \Delta_1^c \vert   - \frac{1}{q}   \sum\limits_{u \in \mathbb{F}^r_p}\sum\limits_{t_1 \in \Delta^c_1}  w_p^{\langle u, \langle \beta_1, t_1\rangle_q\rangle_p}\Big]\\ 
					&~ + 2\vert \Delta^c_2\vert \Big[\vert \Delta_1 \vert   - \frac{1}{q}   \sum\limits_{u \in \mathbb{F}^r_p}\sum\limits_{t_1 \in \Delta_1}  w_p^{\langle u, \langle \beta_1, t_1\rangle_q\rangle_p}\Big]\\
					& = 2 q^{2m-1} (q-1)\big(1 - \delta_{\mathbf{0}, \beta_1}\big) - 2q^m\mathcal{A}_{\Delta_1, \beta_1} + 2 \vert \Delta_2^c\vert \mathcal{A}_{\Delta_1, \beta_1}\\
					& = 2 q^{2m-1} (q-1)\big(1 - \delta_{\mathbf{0}, \beta_1}\big)-2 \vert \Delta_2 \vert \mathcal{A}_{\Delta_1, \beta_1}.
				\end{split}
			\end{equation}
			Now we discuss the following cases:
			\begin{enumerate}[label=(\roman*)]
				\item We have, $wt_{Lee}(c_{D}(v))=0 \iff \beta_1 = \textbf{0}$.\\
				In this case, $\#\beta_1=1, \#\beta_2=q^m.$\\
				Therefore, we have $\#v=q^m.$
				
				\item When $\beta_1 \neq \textbf{0}$, then $wt_{Lee}(c_{D}(v))
				= 2 (q-1)q^{2m-1}  - 2 \vert \Delta_2 \vert \mathcal{A}_{\Delta_1, \beta_1}$. From Lemma \ref{Lem-A}(\ref{Lem-A3}), we have $$\mathcal{A}_{\Delta_1, \beta_1} \leq \sum\limits_{i \in [s]} (q-1)q^{\vert M_i\vert -1}.$$
				Therefore,		 
				$wt_{Lee}(c_{D}(v))
				\geq 2 (q-1)q^{2m-1} - 2 \vert \Delta_2 \vert \sum\limits_{i \in [s]} (q-1)q^{\vert M_i\vert -1}.$ By Lemma \ref{Lem-A}(\ref{Lem-A3}) and Lemma \ref{Lem-A+}(\ref{Lem-A+2}), there exists $\beta_1 \in \mathbb{F}_q^m$ such that 
				\begin{equation*}
					wt_{Lee}(c_{D}(v))
					= 2 (q-1)q^{2m-1} - 2 \vert \Delta_2 \vert \sum\limits_{i \in [s]} (q-1)q^{\vert M_i\vert -1}.
				\end{equation*}
				Hence the minimum distance $d_{L} = 2  (q-1)\big(q^{2m-1}  -  \vert \Delta_2 \vert \sum\limits_{i \in [s]} q^{\vert M_i\vert -1}\big)$. \qed
				
			\end{enumerate}
		\end{proof}

        The following theorem determines the parameters and establishes the minimality conditions for the Gray images of the codes presented in Theorem~\ref{Thm-1}.
		\begin{theorem}\label{Thm-GrayImage}
		Let $m, r, s, k \in \mathbb{N}$ and $q=p^r$. Suppose that $\Delta_1$ and $\Delta_2$ are simplicial complexes of $\mathbb{F}_q^m$ with
		\[
		\mathcal{M}=\{M_1, M_2, \dots, M_s\}
		\quad \text{and} \quad
		\mathcal{N}=\{N_1, N_2, \dots, N_k\}
		\]
		being the sets of maximal elements of $\Delta_1$ and $\Delta_2$, respectively. Let $M_i \setminus \cup_{j\in [s]\setminus \{i\}}M_j\neq \emptyset$ for any $i\in [s]$ and let $\Phi$ be the map given by Eq.~\eqref{PhiEq}.
		\begin{enumerate}
			\item \label{part:1GI}
			Let $D = a\Delta_1 + b\Delta_2 \subseteq \mathcal{S}^m$. Then $\Phi(C_D)$ is an $\big[2 \vert \Delta_1 \vert \vert \Delta_2 \vert, {\vert \cup_{i\in [s]} M_i \vert}, 2\vert \Delta_2\vert (q-1) q^{{\rm min}_{i\in [s]} \{\vert M_i \vert\} -1}\big]$-linear code over $\mathbb{F}_q$, and is minimal if
			\[
			q^{{\rm min}_{i \in [s]}\{\vert M_i \vert\} +1} > (q-1)\sum_{ i \in [s]}q^{\vert M_i \vert}.
			\]
			
			\item \label{part:2GI}
			Let $D = a\Delta_1 + b\Delta^c_2 \subseteq \mathcal{S}^m$. Then $\Phi(C_D)$ is an $\big[2\vert \Delta_1 \vert (q^m- \vert \Delta_2 \vert ), {\vert \cup_{i\in [s]} M_i \vert}, 2(q^m-\vert \Delta_2 \vert )(q-1)q^{{\rm min}_{i\in [s]} \{\vert M_i \vert\} -1} \big]$-linear code over $\mathbb{F}_q$, and is minimal if
			\[
			q^{{\rm min}_{i \in [s]}\{\vert M_i \vert \} +1} > (q-1)\sum_{ i \in [s]}q^{\vert M_i \vert}.
			\]
			
			\item \label{part:3GI}
			Let $D = a\Delta^c_1 + b\Delta_2 \subseteq \mathcal{S}^m$. Then $\Phi(C_D)$ is an $\big[2(q^m - \vert \Delta_1 \vert ) \vert \Delta_2 \vert, m, 2\vert \Delta_2 \vert (q-1)\big(q^{m-1} - \sum_{ i \in [s]}q^{ \vert M_i \vert -1}\big)\big]$-linear code over $\mathbb{F}_q$, and is minimal if
			\[
			q^m > \sum_{ i \in [s]} q^{\vert M_i \vert + 1}.
			\]
			
			\item \label{part:4GI}
			Let $D = a\Delta^c_1 + b\Delta^c_2 \subseteq \mathcal{S}^m$. Then $\Phi(C_D)$ is an $\big[2(q^m- \vert \Delta_1 \vert ) (q^m - \vert \Delta_2 \vert ), m, 2(q^m - \vert \Delta_2 \vert )(q-1)\big(q^{m-1} - \sum_{ i \in [s]} q^{\vert M_i \vert -1}\big)\big]$-linear code over $\mathbb{F}_q$, and is minimal if
			\[q^m > \sum_{ i \in [s]} q^{\vert M_i \vert + 1}.\]
			
			\item \label{part:5GI}
			Let $D = a\Delta_1 + b\Delta_2 \subseteq \mathcal{S}^m$. Then $\Phi(C_{D^c})$ is an $\big[2(q^{2m}-\vert \Delta_1 \vert \vert \Delta_2 \vert), m, 2(q-1)\big(q^{2m-1} - \vert \Delta_2 \vert \sum_{ i \in [s]} q^{\vert M_i \vert -1}\big)\big]$-linear code over $\mathbb{F}_q$, and is minimal if
			\[
			q^{2m} > \vert \Delta_2 \vert \sum_{ i \in [s]} q^{\vert M_i \vert + 1}.
			\]
		\end{enumerate}
	    \end{theorem}
       % \begin{proof}
        %    We give the proof of part \ref{part:5GI} only, as the remaining parts can be proved similarly.\\
         %   \textcolor{red}{Proof pending}
        %\end{proof}

	%%%Section4	
	\section{Lee weight distribution of $C_D$ when each $\Delta_i$ has a single maximal element}\label{Sec:4}
		In this section, we determine the Lee weight distribution of $C_D$, when both $\Delta_1$ and $\Delta_2$ are simplicial complexes with a single maximal element.
		
		\begin{proposition}\label{Prop:2}
			Let $r, m \in \mathbb{N}$, $q=p^r$ and let $M, N\subseteq [m]$. Suppose  $\Delta_M$ and $\Delta_N$ are simplicial complexes of $\mathbb{F}_q^m$ with maximal elements $M$ and $N$, respectively.
			\begin{enumerate}
				\item \label{prop:1}
				Let $D=a\Delta_M + b\Delta_N \subseteq\mathcal{S}^m$. Then $C_D$ is an $\mathcal{S}$-linear code of length $q^{\vert M\vert +\vert N \vert} $, size $q^{\vert M \vert}$ and its Lee weight distribution is presented in Table \ref{Table:1}. 
				
				\begin{table}[h]
					\centering
					%\begin{adjustbox}{width=\textwidth}
					\begin{tabular}{  c | c  }
						\hline
						Lee weight    & Frequency \\
						\hline
						$2(q-1)q^{\vert M \vert + \vert N \vert -1}$ & $q^{2m-\vert M \vert} (q^{\vert M \vert}-1)$\\
						\hline
						$0$ & $q^{2m - \vert M \vert}$\\
						\hline
					\end{tabular}
					%\end{adjustbox}
					\caption{Lee weight distribution in Theorem \ref{Prop:2}(\ref{prop:1})}
					\label{Table:1}
				\end{table}
				
				\item \label{prop:2}
				Let $D=a\Delta_M + b\Delta_N^c \subseteq\mathcal{S}^m$. Then $C_D$ is an $\mathcal{S}$-linear code of length $q^{\vert M\vert}(q^m- q^{\vert N \vert})$, size $q^{\vert M \vert}$ and its Lee weight distribution is presented in Table \ref{Table:2}.
				
				\begin{table}[h]
					\centering
					%\begin{adjustbox}{width=\textwidth}
					\begin{tabular}{  c | c  }
						\hline
						Lee weight    & Frequency \\
						\hline
						$2(q-1)(q^m - q^{\vert N \vert })q^{\vert M \vert -1}$ & $q^{2m-\vert M \vert} (q^{\vert M \vert}-1)$\\
						\hline
						$0$ & $q^{2m - \vert M \vert}$\\
						\hline
					\end{tabular}
					%\end{adjustbox}
					\caption{Lee weight distribution in Theorem \ref{Prop:2}(\ref{prop:2})}
					\label{Table:2}
				\end{table}
				
				\item \label{prop:3}
				Let $D=a\Delta_M^c + b\Delta_N \subseteq\mathcal{S}^m$. Then $C_D$ is an $\mathcal{S}$-linear code of length $q^{\vert N\vert}(q^m- q^{\vert M \vert})$, size $q^{m}$ and its Lee weight distribution is presented in Table \ref{Table:3}.
				
				\begin{table}[h]
					\centering
					%\begin{adjustbox}{width=\textwidth}
					\begin{tabular}{  c | c  }
						\hline
						Lee weight    & Frequency \\
						\hline
						$2(q-1)q^{m + \vert N \vert -1}$ & $q^m(q^{m-\vert M\vert} -1)$\\
						\hline
						$2(q-1)(q^m - q^{\vert M \vert})q^{\vert N \vert -1}$ & $q^{2m-\vert M \vert} (q^{\vert M \vert}-1)$\\
						\hline
						$0$ & $q^{m}$\\
						\hline
					\end{tabular}
					%\end{adjustbox}
					\caption{Lee weight distribution in Theorem \ref{Prop:2}(\ref{prop:3})}
					\label{Table:3}
				\end{table}
				
				\item \label{prop:4}
				Let $D=a\Delta_M^c + b\Delta_N^c \subseteq\mathcal{S}^m$. Then $C_D$ is an $\mathcal{S}$-linear code of length $(q^m - q^{\vert M \vert})(q^m- q^{\vert N \vert})$, size $q^{m}$ and its Lee weight distribution is presented in Table \ref{Table:4}.
				
				\begin{table}[h]
					\centering
					%\begin{adjustbox}{width=\textwidth}
					\begin{tabular}{  c | c  }
						\hline
						Lee weight    & Frequency \\
						\hline
						$2(q-1)(q^m - q^{\vert N \vert})q^{m -1}$ & $q^m(q^{m-\vert M\vert} -1)$\\
						\hline
						$2(q-1)(q^m - q^{\vert N \vert})(q^{m -1} - q^{\vert M \vert -1})$ & $q^{2m-\vert M \vert} (q^{\vert M \vert}-1)$\\
						\hline
						$0$ & $q^{m}$\\
						\hline
					\end{tabular}
					%\end{adjustbox}
					\caption{Lee weight distribution in Theorem \ref{Prop:2}(\ref{prop:4})}
					\label{Table:4}
				\end{table}
				
				\item \label{prop:5}
				Let $D=a\Delta_M + b\Delta_N \subseteq\mathcal{S}^m$. Then $C_{D^c}$ is an $\mathcal{S}$-linear code of length $q^{2m} - q^{\vert M \vert + \vert N \vert}$, size $q^{m}$ and its Lee weight distribution is presented in Table \ref{Table:5}.
				
				\begin{table}[h]
					\centering
					%\begin{adjustbox}{width=\textwidth}
					\begin{tabular}{  c | c  }
						\hline
						Lee weight    & Frequency \\
						\hline
						$2(q-1)q^{2m -1}$ & $q^m(q^{m-\vert M\vert} -1)$\\
						\hline
						$2(q-1)(q^{2m -1} - q^{\vert M \vert + \vert N \vert -1})$ & $q^{2m-\vert M \vert} (q^{\vert M \vert}-1)$\\
						\hline
						$0$ & $q^{m}$\\
						\hline
					\end{tabular}
					%\end{adjustbox}
					\caption{Lee weight distribution in Theorem \ref{Prop:2}(\ref{prop:5})}
					\label{Table:5}
				\end{table}
				
			\end{enumerate}
		\end{proposition}
		\begin{proof}
			We present the proof of part \eqref{prop:5} only; the proofs of the other parts are similar and therefore omitted. Clearly, the length of the code $C_D$ is $q^{2m}-q^{\vert M \vert + \vert N \vert }$.\\ 
			Let $v = a\beta_1 + b\beta_2 \in \mathcal{S}^m$ and $x=at_1 + bt_2 \in D^c$, where
			\begin{equation*}
				D^c = \big(a\Delta_M^c + b\mathbb{F}_q^m \big) \bigsqcup \big(a\Delta_M + b\Delta_N^c \big).
			\end{equation*}
			From Eq.~\eqref{Eq.13} and Lemma~\ref{Lem-A}, we get
			\begin{equation*}
				wt_{Lee}(c_D(v)) = 2q^{2m-1}(q-1)(1-\delta_{\mathbf{0}, \beta_1})-2\vert \Delta_N \vert \mathcal{A}_{\Delta_M, \beta_1}.
			\end{equation*}
			Now we consider the following cases:
			\begin{enumerate}[label=(\roman*)]
				\item \label{Case-1}
				Let $\beta_1 = \mathbf{0}$. Then $wt_{Lee}(c_D(v)) = 0$. In this case, we have $\# \beta_1 =1, \ \# \beta_2 = q^m$. Therefore, $\# v=q^m.$
				
				\item Let $\beta_1 \neq \mathbf{0}$. Then we have
				\begin{equation*}
					\begin{split}
						wt_{Lee}(c_D(v)) &= 2q^{2m-1}(q-1)-2\vert \Delta_N \vert \mathcal{A}_{\Delta_M, \beta_1}\\
						&=2(q-1)q^{2m-1}-2(q-1)q^{\vert M \vert + \vert N \vert -1}\big(1-\Psi(\beta_1 \mid \Delta_M)\big).
					\end{split}
				\end{equation*}
				Here we deal the following sub-cases:
				\begin{enumerate}[label=(\alph*)]
					\item If $\beta_1 \in \Delta_M^{\perp}$ then $wt_{Lee}(c_D(v)) =2(q-1)q^{2m-1}$. In this case, we have $\# \beta_1 = q^{m-\vert M \vert }-1, \ \# \beta_2 = q^m$. Therefore, $\# v=q^m(q^{m-\vert M \vert }-1).$
					
					\item If $\beta_1 \notin \Delta_M^{\perp}$ then $wt_{Lee}(c_D(v)) = 2(q-1)\big[q^{2m-1}-q^{\vert M \vert + \vert N \vert -1}\big]$. In this case, we have $\# \beta_1 = q^m - q^{m-\vert M \vert }, \ \# \beta_2 = q^m$. Therefore, $\# v=q^{2m - \vert M \vert }(q^{\vert M \vert} - 1).$
				\end{enumerate}
			\end{enumerate}
			Based on the above calculations, we obtain Table \ref{Table:5}.\\
			Define the map $c_D : \mathcal{S}^m \longrightarrow C_D$ defined by $c_D(v) = (v \cdot d)_{d \in D}$. Note that it is a surjective linear transformation. From case \ref{Case-1}, we have $|\ker(c_D)| = |\{ v \in \mathcal{S}^m : x \cdot d = 0 \ \forall \ d \in D \}| = q^{m}$. By the first isomorphism theorem, we obtain $\vert C_D\vert = \frac{\vert \mathcal{S}^m \vert }{\vert \ker(c_D) \vert} = q^{m}$. This completes the proof.~\qed 
		\end{proof}
		
		\begin{theorem}
			Let $r, m \in \mathbb{N}$, $q=p^r$ and let $M, N\subseteq [m]$. Suppose  $\Delta_M$ and $\Delta_N$ are simplicial complexes of $\mathbb{F}_q^m$ with maximal elements $M$ and $N$, respectively.
			\begin{enumerate}
				\item \label{corol:1}
				Let $D=a\Delta_M + b\Delta_N \subseteq\mathcal{S}^m$. Then $\Phi(C_D)$ is a $1$-weight $2(q-1)q^{\vert M \vert + \vert N \vert -1}$-divisible code over $\mathbb{F}_q$ with parameters $\big[2q^{\vert M\vert +\vert N \vert}, \vert M \vert, 2(q-1)q^{\vert M \vert + \vert N \vert -1}\big]$, and hence minimal.
				
				\item \label{corol:2}
				Let $D=a\Delta_M + b\Delta_N^c \subseteq\mathcal{S}^m$. Then $\Phi(C_D)$ is a $1$-weight $2(q-1)(q^m - q^{\vert N \vert })q^{\vert M \vert -1}$-divisible code over $\mathbb{F}_q$ with parameters $\big[2q^{\vert M\vert}(q^m- q^{\vert N \vert}), \vert M \vert, 2(q-1)(q^m - q^{\vert N \vert })q^{\vert M \vert -1} \big]$, and hence minimal.
				
				\item \label{corol:3}
				Let $D=a\Delta_M^c + b\Delta_N \subseteq\mathcal{S}^m$. Then $\Phi(C_D)$ is a $2$-weight $2(q-1)q^{\vert M \vert + \vert N \vert -1}$-divisible code over $\mathbb{F}_q$ with parameters $\big[2q^{\vert N\vert}(q^m- q^{\vert M \vert}), m,  2(q-1)(q^m - q^{\vert M \vert})q^{\vert N \vert -1} \big]$, where $wt_{0} = 2(q-1)(q^m - q^{\vert M \vert})q^{\vert N \vert -1}$ and $wt_{\infty} = 2(q-1)q^{m + \vert N \vert -1}$, and is minimal if $m > {\vert M \vert +1}$.
                \begin{enumerate}
                    \item Let $\vert M \vert + \vert N \vert \leq m$ and $\theta_1 = 2(q^{\vert N \vert }-1)$. If $0 \leq \theta_1 < 2(q-1)(\vert M \vert + \vert N \vert)$ then $\Phi(C_D)$ is optimal with respect to the Griesmer bound.

                    \item Let $\vert M \vert + \vert N \vert \geq m$ and $\theta_2 = 2(q^{\vert N \vert } - q^{\vert M \vert + \vert N \vert -m})$. If $0 \leq \theta_2 < 2(q-1)m$ then $\Phi(C_D)$ is optimal with respect to the Griesmer bound.
                \end{enumerate}

				\item \label{corol:4}
				Let $D=a\Delta_M^c + b\Delta_N^c \subseteq\mathcal{S}^m$. Then $\Phi(C_D)$ is a $2$-weight $2(q-1)(q^m - q^{\vert N \vert}) q^{\vert M \vert -1}$-divisible code over $\mathbb{F}_q$ with parameters $\big[2(q^m - q^{\vert M \vert})(q^m- q^{\vert N \vert}), m, 2(q-1)(q^m - q^{\vert N \vert})(q^{m -1} - q^{\vert M \vert -1})\big]$, where $wt_{0} = 2(q-1)(q^m - q^{\vert N \vert})(q^{m -1} - q^{\vert M \vert -1})$ and $wt_{\infty} = 2(q-1)(q^m - q^{\vert N \vert})q^{m -1}$, and is minimal if $m > {\vert M \vert +1}$.
				\begin{enumerate}
                    \item Let $\vert M \vert + \vert N \vert \leq m$ and $\theta_1 = 2[(q^m - q^{\vert M \vert} - q^{\vert N \vert} + 1) -(q-1)(m-\vert M \vert - \vert N \vert )]$. If $0 \leq \theta_1 < 2(q-1)(\vert M \vert + \vert N \vert)$ then $\Phi(C_D)$ is optimal with respect to the Griesmer bound.

                    \item Let $\vert M \vert + \vert N \vert \geq m$ and $\theta_2 = 2(q^m - q^{\vert M \vert} - q^{\vert N \vert } + q^{\vert M \vert + \vert N \vert -m})$. If $0 \leq \theta_2 < 2(q-1)m$ then $\Phi(C_D)$ is optimal with respect to the Griesmer bound.
                \end{enumerate}

				\item \label{corol:5}
				Let $D=a\Delta_M + b\Delta_N \subseteq\mathcal{S}^m$. Then $\Phi(C_{D^c})$ is a $2$-weight $2(q-1) q^{\vert M \vert + \vert N \vert -1}$-divisible code over $\mathbb{F}_q$ with parameters $\big[2(q^{2m} - q^{\vert M \vert + \vert N \vert}), m, 2(q-1)(q^{2m -1} - q^{\vert M \vert + \vert N \vert -1}) \big]$, where $wt_{0} = 2(q-1)(q^{2m -1} - q^{\vert M \vert + \vert N \vert -1})$ and $wt_{\infty} = 2(q-1)q^{2m -1}$, and is minimal if $2m > \vert M \vert + \vert N \vert +1$.
                \begin{enumerate}
                    \item Let $\vert M \vert + \vert N \vert \leq m$ and $\theta_1 = 2(q^{m}-1)$. If $0 \leq \theta_1 < 2(q-1)(\vert M \vert + \vert N \vert)$ then $\Phi(C_D)$ is optimal with respect to the Griesmer bound.

                    \item Let $\vert M \vert + \vert N \vert \geq m$ and $\theta_2 = 2(q^{m} - q^{\vert M \vert + \vert N \vert -m})$. If $0 \leq \theta_2 < 2(q-1)m$ then $\Phi(C_D)$ is optimal with respect to the Griesmer bound.
                \end{enumerate}
			\end{enumerate}
		\end{theorem}

%Section 5
	\section{Construction of subfield-like codes from $\mathcal{S}$-linear codes}\label{Sec:5}
		In this section, we present the construction of subfield-like codes over $\mathbb{F}_q$ from $\mathcal{S}$-linear codes ${C}_{D}$ by means of an $\mathbb{F}_q$-linear map. For more details, we refer the reader to~\cite{AMC1}.
		
		Define the projection $\tau : \mathcal{S} \longrightarrow \mathbb{F}_q$ by
		\begin{equation}\label{tauMap}
			\tau(a\beta_1 + b\beta_2) = \beta_2.
		\end{equation}
		We extend this to a map from $\mathcal{S}^m$ to $\mathbb{F}_q^{m}$ component-wise for any $m \in \mathbb{N}$. Note that $\tau: C_{D} \longrightarrow \tau(C_{D})\subseteq \mathbb{F}_{q}^{\vert D\vert}$ defined by
		\begin{equation}
			\tau\big(c_{D}(v)\big) = \big(\tau(v\cdot d)\big)_{d\in D}
		\end{equation}
		is a surjective linear transformation of $\mathbb{F}_q$-vector spaces and $\tau(C_{D})$ is a linear code over $\mathbb{F}_q$ having length $\vert D\vert$.	
		
		Let $v = a\beta_1 + b\beta_2 \in \mathcal{S}^m$ and $D = aD_1 + bD_2 \subseteq \mathcal{S}^m$. By Eq. \eqref{c_DMap}, we have
		\begin{equation}\label{TauDef}
			\tau(C_{D}) = \{\tau(c_{D}(v)) = (\tau(v\cdot d))_{d\in D} : v \in \mathcal{S}^m \}\subseteq \mathbb{F}_{q}^{\vert D\vert}.
		\end{equation}
		Note that $\tau(C_{D})$ is an $\mathbb{F}_q$-linear code of length $\vert D \vert $.
		Here we define $T := \tau \circ c_{D}$, where $c_{D}$ is defined by Eq. \eqref{c_DMap} and $\tau$ is defined by Eq. \eqref{tauMap}.
		Observe that the map $T: \mathcal{S}^m \longrightarrow C_{D} \longrightarrow \tau(C_D) \subseteq \mathbb{F}_{q}^{\vert D \vert}$ defined by
		\begin{equation}
			T(x) = \big(\tau(c_{D}(x))\big) = \big(\tau(x\cdot d)\big)_{d\in D}
		\end{equation}
		is a surjective linear transformation of $\mathbb{F}_q$-vector spaces. By the first isomorphism theorem, we have
		\begin{equation}\label{DimCount}
			\vert \tau(C_{D})\vert = \frac{\vert \mathcal{S}^m \vert}{\vert \ker(T)\vert}.
		\end{equation}
		
		To obtain the weight distribution of $\tau(C_D)$, we first determine the Hamming weight of an arbitrary element of $\tau(C_D)$. Let $v = a\beta_1 + b\beta_2 \in \mathcal{S}^m$ and $d = at_1 + bt_2 \in D = aD_1 + bD_2 \subseteq \mathcal{S}^m$. We have
		\begin{equation}\label{KeyEq_SF_like}
			\begin{split}
				{ wt_H}(\tau (c_D(v)))
				& = { wt_H}(\tau ((a\beta_1 + b\beta_2)(at_1 + bt_2))_{t_i\in D_i})\\
				& = { wt_H}(\tau ((b \langle \beta_1, t_1 \rangle_q))_{t_i\in D_i})\\
				&={ wt_H}((\langle \beta_1, t_1 \rangle_q)_{t_i\in D_i}).
			\end{split}
		\end{equation}
		Observe from Eq.~\eqref{Eq-Lee} that
		\[
		wt_H\bigl(\tau(c_D(v))\bigr) =\frac{1}{2}\,wt_{Lee}\bigl(c_D(v)\bigr).
		\]
		Therefore, the proof of the following result follows directly from that of Theorem~\ref{Thm-1}.
		
		\begin{theorem}\label{Thm-2}
			Let $m, r, s, k \in \mathbb{N}$ and $q=p^r$. Suppose that $\Delta_1$ and $\Delta_2$ are simplicial complexes of $\mathbb{F}_q^m$ with
			\[
			\mathcal{M}=\{M_1,M_2,\dots,M_s\}
			\quad \text{and} \quad
			\mathcal{N}=\{N_1,N_2,\dots,N_k\}
			\]
			being the sets of maximal elements of $\Delta_1$ and $\Delta_2$, respectively. Let $M_i \setminus \cup_{j\in [s]\setminus \{i\}}M_j\neq \emptyset$ for any $i\in [s]$. 
			\begin{enumerate}
				\item \label{part:1sub}
				Let $D = a\Delta_1 + b\Delta_2 \subseteq \mathcal{S}^m$. Then $\tau{(C_D)}$ is an $\big[\vert \Delta_1 \vert \vert \Delta_2 \vert, \vert \cup_{i\in [s]} M_i \vert, \vert \Delta_2\vert (q-1) q^{{\rm min}_{i\in [s]} \{\vert M_i \vert \} -1}\big]$-linear code over $\mathbb{F}_q$, and is minimal if
                \[
                q^{{\rm min}_{i \in [s]} \{\vert M_i \vert} + 1 \} > (q-1)\sum_{ i \in [s]}q^{\vert M_i \vert}.
                \]
				
				\item \label{part:2sub}
				Let $D = a\Delta_1 + b\Delta_2^c \subseteq \mathcal{S}^m$. Then $\tau{(C_D)}$ is an $\big[\vert \Delta_1 \vert (q^m- \vert \Delta_2 \vert), \vert \cup_{i\in [s]} M_i \vert, (q^m - \vert \Delta_2\vert) (q-1) q^{{\rm min}_{i\in [s]} \{\vert M_i \vert \} -1}\big]$-linear code over $\mathbb{F}_q$, and is minimal if
                \[
                q^{{\rm min}_{i \in [s]} \{\vert M_i \vert} + 1 \} > (q-1)\sum_{ i \in [s]}q^{\vert M_i \vert}.
                \]

				\item \label{part:3sub}
				Let $D = a\Delta_1^c + b\Delta_2 \subseteq \mathcal{S}^m$. Then $\tau{(C_D)}$ is an $\big[ (q^m - \vert \Delta_1 \vert) \vert \Delta_2 \vert, m, \vert \Delta_2\vert (q-1)\big(q^{m-1} - \sum_{i\in [s]}q^{ \vert M_i \vert -1}\big)\big]$-linear code over $\mathbb{F}_q$, and is minimal if
                \[
                q^m > \sum_{ i \in [s]} q^{\vert M_i \vert + 1}.
                \]
				
				\item \label{part:4sub}
				Let $D = a\Delta_1^c + b\Delta_2^c \subseteq \mathcal{S}^m$. Then $\tau{(C_D)}$ is an $\big[ (q^m - \vert \Delta_1 \vert) (q^m- \vert \Delta_2 \vert), m, (q^m - \vert \Delta_2\vert) (q-1)\big(q^{m-1} - \sum_{i\in [s]} q^{ \vert M_i \vert -1}\big)\big]$-linear code over $\mathbb{F}_q$, and is minimal if
                \[
                q^m > \sum_{ i \in [s]} q^{\vert M_i \vert + 1}.
                \]
				
				\item \label{part:5sub}
				Let $D = a\Delta_1 + b\Delta_2 \subseteq \mathcal{S}^m$. Then $\tau{(C_{D^c})}$ is an $\big[ (q^{2m} - \vert \Delta_1 \vert \vert \Delta_2 \vert), m, (q-1)\big(q^{2m-1} - \vert \Delta_2 \vert \sum_{i\in [s]} q^{ \vert M_i \vert -1}\big)\big]$-linear code over $\mathbb{F}_q$, and is minimal if
                \[
                q^{2m} > \vert \Delta_2 \vert \sum_{ i \in [s]} q^{\vert M_i \vert + 1}.
                \]
			\end{enumerate}
		\end{theorem}
        
        \begin{corollary}
            Let $r, m \in \mathbb{N}$, $q=p^r$ and let $M, N\subseteq [m]$. Suppose  $\Delta_M$ and $\Delta_N$ are simplicial complexes of $\mathbb{F}_q^m$ with maximal elements $M$ and $N$, respectively.
			\begin{enumerate}
				\item \label{corol:1SF}
				Let $D=a\Delta_M + b\Delta_N \subseteq\mathcal{S}^m$. Then $\tau(C_D)$ is a $1$-weight $(q-1)q^{\vert M \vert + \vert N \vert -1}$-divisible code over $\mathbb{F}_q$ with parameters $\big[q^{\vert M\vert +\vert N \vert}, \vert M \vert, (q-1)q^{\vert M \vert + \vert N \vert -1}\big]$, and hence minimal.
				
				\item \label{corol:2SF}
				Let $D=a\Delta_M + b\Delta_N^c \subseteq\mathcal{S}^m$. Then $\tau(C_D)$ is a $1$-weight $(q-1)(q^m - q^{\vert N \vert })q^{\vert M \vert -1}$-divisible code over $\mathbb{F}_q$ with parameters $\big[q^{\vert M\vert}(q^m- q^{\vert N \vert}), \vert M \vert, (q-1)(q^m - q^{\vert N \vert })q^{\vert M \vert -1} \big]$, and hence minimal.
				
				\item \label{corol:3SF}
				Let $D=a\Delta_M^c + b\Delta_N \subseteq\mathcal{S}^m$. Then $\tau(C_D)$ is a $2$-weight $(q-1)q^{\vert M \vert + \vert N \vert -1}$-divisible code over $\mathbb{F}_q$ with parameters $\big[(q^m- q^{\vert M \vert})q^{\vert N\vert}, m,  (q-1)(q^m - q^{\vert M \vert})q^{\vert N \vert -1} \big]$, where $wt_{0} = (q-1)(q^m - q^{\vert M \vert})q^{\vert N \vert -1}$ and $wt_{\infty} = (q-1)q^{m + \vert N \vert -1}$, and is minimal if $m > {\vert M \vert +1}$.
                \begin{enumerate}
                    \item Let $\vert M \vert + \vert N \vert \leq m$ and $\theta_1 = (q^{\vert N \vert }-1)$. If $0 \leq \theta_1 < (q-1)(\vert M \vert + \vert N \vert)$ then $\tau(C_D)$ is optimal with respect to the Griesmer bound.

                    \item Let $\vert M \vert + \vert N \vert \geq m$ and $\theta_2 = (q^{\vert N \vert } - q^{\vert M \vert + \vert N \vert -m})$. If $0 \leq \theta_2 < (q-1)m$ then $\tau(C_D)$ is optimal with respect to the Griesmer bound.
                \end{enumerate}

				\item \label{corol:4SF}
				Let $D=a\Delta_M^c + b\Delta_N^c \subseteq\mathcal{S}^m$. Then $\tau(C_D)$ is a $2$-weight $(q-1)(q^m - q^{\vert N \vert}) q^{\vert M \vert -1}$-divisible code over $\mathbb{F}_q$ with parameters $\big[(q^m - q^{\vert M \vert})(q^m- q^{\vert N \vert}), m, (q-1)(q^m - q^{\vert N \vert})(q^{m -1} - q^{\vert M \vert -1})\big]$, where $wt_{0} = (q-1)(q^m - q^{\vert N \vert})(q^{m -1} - q^{\vert M \vert -1})$ and $wt_{\infty} = (q-1)(q^m - q^{\vert N \vert})q^{m -1}$, and is minimal if $m > {\vert M \vert +1}$.
                \begin{enumerate}
                    \item Let $\vert M \vert + \vert N \vert \leq m$ and $\theta_1 = [(q^m - q^{\vert M \vert} - q^{\vert N \vert} + 1) -(q-1)(m-\vert M \vert - \vert N \vert )]$. If $0 \leq \theta_1 < (q-1)(\vert M \vert + \vert N \vert)$ then $\tau(C_D)$ is optimal with respect to the Griesmer bound.

                    \item Let $\vert M \vert + \vert N \vert \geq m$ and $\theta_2 = (q^m - q^{\vert M \vert} - q^{\vert N \vert } + q^{\vert M \vert + \vert N \vert -m})$. If $0 \leq \theta_2 < (q-1)m$ then $\tau(C_D)$ is optimal with respect to the Griesmer bound.
                \end{enumerate}

				\item \label{corol:5SF}
				Let $D=a\Delta_M + b\Delta_N \subseteq\mathcal{S}^m$. Then $\tau(C_{D^c})$ is a $2$-weight $(q-1) q^{\vert M \vert + \vert N \vert -1}$-divisible code over $\mathbb{F}_q$ with parameters $\big[(q^{2m} - q^{\vert M \vert + \vert N \vert}), m, (q-1)(q^{2m -1} - q^{\vert M \vert + \vert N \vert -1}) \big]$, where $wt_{0} = (q-1)(q^{2m -1} - q^{\vert M \vert + \vert N \vert -1})$ and $wt_{\infty} = (q-1)q^{2m -1}$, and is minimal if $2m > \vert M \vert + \vert N \vert +1$.
                \begin{enumerate}
                    \item Let $\vert M \vert + \vert N \vert \leq m$ and $\theta_1 = (q^{m}-1)$. If $0 \leq \theta_1 < (q-1)(\vert M \vert + \vert N \vert)$ then $\tau(C_D)$ is optimal with respect to the Griesmer bound.

                    \item Let $\vert M \vert + \vert N \vert \geq m$ and $\theta_2 = (q^{m} - q^{\vert M \vert + \vert N \vert -m})$. If $0 \leq \theta_2 < (q-1)m$ then $\tau(C_D)$ is optimal with respect to the Griesmer bound.
                \end{enumerate}
			\end{enumerate}
        \end{corollary}

%%%Section 6
		\section{Self-orthogonal codes}\label{Sec:6}
       In this section, we establish certain conditions under which the Gray image code $\Phi(C_D)$ and the subfield-like code $\tau(C_D)$ are self-orthogonal, where $C_D$ is an $\mathcal{S}$-linear code.

		Let $C$ be an $[n,k,d]$-linear code over $\mathbb{F}_q$. The Euclidean dual of $C$ is defined by
		\[
		C^\perp
		=
		\{x\in \mathbb{F}_q^n : \langle x,c\rangle_q=0 \text{ for all } c\in C\}.
		\]
		If $C\subseteq C^\perp$, then $C$ is called \textit{Euclidean self-orthogonal}.

        Let $C$ be an $[n,k,d]$-linear code over $\mathbb{F}_q$. Then $C$ is called \emph{projective} if the dual code $C^\perp$ has minimum Hamming distance at least $3$.
		
		When the field is a quadratic extension, that is, $\mathbb{F}_{q^2}$, we consider the Hermitian inner product on $\mathbb{F}_{q^2}^n$ defined by
		\begin{equation*}
			\langle x,y\rangle_q^H
			=
			\langle x,\bar{y}\rangle_q
			=
			\sum_{i=1}^{n}x_i\bar{y}_i,
		\end{equation*}
		where $\bar{y}_i$ denotes the conjugate of $y_i$. The Hermitian dual of a quaternary code $C$ is defined as
		\begin{equation*}
			C^{\perp_H}
			=
			\{x\in \mathbb{F}_{q^2}^n :
			\langle x,c\rangle_{q}^H=0 \text{ for all } c\in C\}.
		\end{equation*}
		If $C\subseteq C^{\perp_H}$, then $C$ is called \textit{Hermitian self-orthogonal}.\\
		The following lemma provides a useful criterion for determining whether a code is Euclidean or Hermitian self-orthogonal from its weight distribution.
		\begin{theorem}[\cite{wchuffman}]\label{Thm-4}
			Let $C$ be a linear code over $\mathbb{F}_q$, where $q=2,3$, or $4$.
			\begin{enumerate}
				\item If $q=2$, then $C$ is Euclidean self-orthogonal whenever every codeword of $C$ has weight divisible by $4$.
				
				\item If $q=3$, then $C$ is Euclidean self-orthogonal if and only if every codeword of $C$ has weight divisible by $3$.
				
				\item If $q=4$, then $C$ is Hermitian self-orthogonal if and only if every codeword of $C$ has even weight.
			\end{enumerate}
		\end{theorem}
		
		By using the above lemma, we obtain the following result from Theorem~\ref{Thm-1} and Theorem~\ref{Thm-2}.
		
		\begin{theorem}\label{th5}
			Let $m, r, s, k \in \mathbb{N}$ and $q=p^r$. Suppose that $\Delta_1$ and $\Delta_2$ are simplicial complexes of $\mathbb{F}_q^m$ with
			\[
			\mathcal{M}=\{M_1,M_2,\dots,M_s\}
			\quad \text{and} \quad
			\mathcal{N}=\{N_1,N_2,\dots,N_k\}
			\]
			being the sets of maximal elements of $\Delta_1$ and $\Delta_2$, respectively. Define
			\begin{equation}\label{EqTheta}
				\theta: = {\rm min}\{ \vert \cap S_1 \vert + \vert \cap S_2 \vert - 1 :  \emptyset \neq S_1 \subseteq \mathcal{M}, \emptyset \neq S_2 \subseteq \mathcal{N} \}.
			\end{equation}
			Then we have the following result.
			\begin{enumerate}
				\item \label{SO-1}
				Let $C_D$ be the linear codes considered in Theorem \ref{Thm-1} and let $\Phi$ be the map defined by Eq.~\eqref{PhiEq}.
				\begin{enumerate}[label=(\roman*)]
					\item If $q=2$ or $3$ and $\theta\geq 1$, then the code $\Phi(C_D)$ is Euclidean self-orthogonal.
					
					\item If $q=4$ and $\theta\geq 0$, then the code $\Phi(C_D)$ is Hermitian self-orthogonal.
				\end{enumerate}
				
				\item \label{SO-2}
				Let $\tau(C_D)$ be the linear codes considered in Theorem \ref{Thm-2}, where $\tau$ is the map defined by Eq.~\eqref{tauMap}.
				\begin{enumerate}[label=(\roman*)]
					\item If $q=2$ and $\theta \geq 2$, then the code $\tau(C_D)$ is Euclidean self-orthogonal. Moreover, if $q=3$ and $\theta \geq 1$, then the code $\tau(C_D)$ is Euclidean self-orthogonal.
					
					\item If $q=4$ and $\theta \geq 1$, then the code $\tau(C_D)$ is Hermitian self-orthogonal.
				\end{enumerate}
			\end{enumerate}
		\end{theorem}
		\begin{proof}
			We will prove part \ref{SO-1} for Theorem~\ref{Thm-1}(\ref{part:5}) only. The proofs for the part \eqref{SO-2} can be obtained by similar arguments.
			Let $v = a\beta_1 + b\beta_2 \in \mathcal{S}^m$. By using Eq. \eqref{Eq.13} and Lemmas \ref{generatinglemma} and \ref{Lem-A}, we have
			\begin{equation*}
				\begin{split}
					wt_{Lee}(c_{D}(v)) 
					= & 2 q^{2m-1} (q-1)\big(1 - \delta_{0, \beta_1}\big)-2 \vert \Delta_2 \vert \mathcal{A}_{\Delta_1, \beta_1}\\
					= & 2 q^{2m-1} (q-1)\big(1 - \delta_{0, \beta_1}\big)\\ 
					& -2\sum_{\substack{\emptyset \neq S_1 \subseteq \mathcal{M}\\ \emptyset \neq S_2 \subseteq \mathcal{N}}} (-1)^{\vert S_1 \vert + \vert S_2 \vert } q^{\vert \cap S_1 \vert + \vert \cap S_2 \vert -1} (q-1)(1 - \Psi(\beta_1 \mid \Delta_{\cap S_1})).
				\end{split}
			\end{equation*}
			Now we consider the following two cases:
			\begin{enumerate}[label=(\roman*)]
				\item If $\beta_1 = \textbf{0}$ then $	wt_{Lee}(c_{D}(v)) =0$.
				\item If $\beta_1 \neq \textbf{0}$ then 
				\begin{equation*}
					wt_{Lee}(c_{D}(v))=2 q^{2m-1} (q-1) -2\sum_{\substack{\emptyset \neq S_1 \subseteq \mathcal{M}\\ \emptyset \neq S_2 \subseteq \mathcal{N}}} (-1)^{\vert S_1 \vert + \vert S_2 \vert } q^{\vert \cap S_1 \vert + \vert \cap S_2 \vert -1} (q-1)(1 - \Psi(\beta_1 \mid \Delta_{\cap S_1})).
				\end{equation*}		
			\end{enumerate}
			Let $\theta: = {\rm min}\{ \vert \cap S_1 \vert + \vert \cap S_2 \vert - 1 :  \emptyset \neq S_1 \subseteq \mathcal{M}, \emptyset \neq S_2 \subseteq \mathcal{N} \}.$
			Now, if $q\in\{2,3\}$ and $\theta\geq 1$, then Theorem~\ref{Thm-4} implies that $\Phi(C_D)$ is Euclidean self-orthogonal. Moreover, if $q=4$ and $\theta\geq 0$, then Theorem~\ref{Thm-4} implies that $\Phi(C_D)$ is Hermitian self-orthogonal. \qed
		\end{proof}

%Section 7

\section{Applications}\label{Sec:7}
        The study of linear codes with specific combinatorial and algebraic properties has become an active area of research because of their numerous applications. In this section, we investigate some of these important classes of linear codes.
        \subsection{Locally recoverable codes}
        %%%%%\textcolor{red}{Check for optimality by using various LRC bounds}\\
        Locally recoverable codes are especially useful in distributed storage systems because a lost symbol can be recovered by accessing only a few other symbols. This reduces the amount of data that needs to be read during the repair process and improves the efficiency of the system. Due to these advantages, locally recoverable codes have received significant attention in recent years.

        A code over $\mathbb{F}_q$ is said to have \emph{locality} $r$ if every coordinate can be recovered by accessing at most $r$ other coordinates, and $r$ is minimal with this property.
        
        \begin{definition}[\cite{LRC-ref}]\label{LRC-Def}
        Let $C$ be a linear code over $\mathbb{F}_q$ with generator matrix $G$. The code $C$ is called a \emph{locally recoverable code} (LRC) with locality $r$ if every column of $G$ can be expressed as an $\mathbb{F}_q$-linear combination of at most $r$ other columns of $G$, where $r$ is the smallest positive integer satisfying this property.
        \end{definition}

        We take $D_2 = \{\mathbf{0}\}$ and consider an ordered defining set $D = aD_1 \subseteq \mathcal{S}^m$, where $D_1 \subseteq \mathbb{F}_q^m$. Now we define the following ordered multiset
        $$D^{(q)}=\{\{ (t, \mathbf{0})  \ \vert \ at  \in D\}\} \subseteq (\mathbb{F}_{q}^m)^2.$$
        Assume that $\vert D^{(q)}\vert = \vert D_1\vert = n$. Here we construct a matrix of size $2m \times n$ over $\mathbb{F}_q$ by using the ordered set $D^{(q)}$ as follows:
        \begin{equation}\label{GenMatrix}
        G = \big[ d_1^{\rm T} \ d_2^{\rm T} \ \cdots \ d_n^{\rm T} \big],
        \end{equation}
        where ${\rm T}$ stands for the transpose operation and $d_i \in D^{(q)}, 1\leq i\leq n$. Clearly, the linear codes over $\mathbb{F}_q$ defined by Eq.~\eqref{TauDef} in Section \ref{Sec:5} is        
        \begin{equation}
        \begin{split}
        \tau(C_D) = C_{D^{(q)}}
                & = \big\{ \big( \langle y,  d_i \rangle_q\big)_{d_i\in D^{(q)}} \ \vert \ y = (\beta_1, \beta_2) \in (\mathbb{F}_q^m)^2 \big\}\\
                &= \big\{yG \ \vert \ y = (\beta_1, \beta_2)\in (\mathbb{F}_q^m)^2 \big\}.
        \end{split}
        \end{equation}
        Therefore, $G$ is a generator matrix of the $\mathbb{F}_q$-linear code $\tau(C_D)$.\par

Let $\Delta^\ast$ denote the set of nonzero elements of a simplicial complex
$\Delta\subseteq\mathbb{F}_q^m$, and let $\overline{\Delta^\ast}$ be the complete set of representatives of the equivalence classes of $\Delta^\ast$ under the relation
\[
x\sim y
\iff
x=\alpha y,\qquad \alpha\in\mathbb{F}_q^\ast.
\]

For each $x\in\Delta^\ast$, its orbit is $[x]=\{\alpha x:\alpha\in\mathbb{F}_q^\ast\}.$ Since $x\neq \mathbf{0}$, every orbit has cardinality $q-1$. As $\overline{\Delta^\ast}$ contains exactly one representative from each orbit, every $x\in\Delta^\ast$ can be uniquely written as $x=\alpha \bar{x},$ for some $\alpha\in\mathbb{F}_q^\ast$ and $\bar{x}\in\overline{\Delta^\ast}$. Therefore, $\Delta^\ast
=
\bigsqcup_{\alpha\in\mathbb{F}_q^\ast}
\alpha\,\overline{\Delta^\ast},$
and hence $|\overline{\Delta^\ast}| = \frac{|\Delta^\ast|}{q-1}.$

The following result holds for $q=2$.
%and follows immediately from Definition~\ref{LRC-Def}.
		\begin{theorem}
			Let $m \in \mathbb{N}$, \ $q=2$ and $M\subseteq [m]$. Suppose $\Delta_{M}$ is the simplicial complex of $\mathbb{F}_2^m$ generated by $M$. Then we have the following:
			\begin{enumerate}
				\item Suppose that $m-2 \geq \vert M \vert \geq 2$ and $D \in \big\{ a\overline{\Delta_M^{\ast}}, a\overline{\Delta_M^{c}}\big\}$, then $\tau(C_D)$ is an LRC with locality $2$.
				
				\item If $D = a\overline{\Delta_M^{c}}$ and $\vert M \vert = m-1$ then $\tau(C_D)$ is an LRC with locality $3$.
			\end{enumerate}
		\end{theorem}
		
		For the case $p\neq 2$, we obtain the following result.
        %, whose proof follows directly from Definition~\ref{LRC-Def}.
		\begin{theorem}
		Let $m, r, s \in \mathbb{N}$ and $q=p^r$, where $p\neq 2$. Suppose $\Delta_1$ is a simplicial complex of $\mathbb{F}_q^m$ with
		\[
			\mathcal{M}=\{M_1,M_2,\dots,M_s\},
		\]
		the set of maximal elements of $\Delta_1$, such that each $\vert M_i \vert \geq 2$. If $D \in \big\{ a\overline{\Delta_1^{\ast}}, a\overline{\Delta_1^{c}}\big\}$ then $\tau(C_D)$ is an LRC with locality $2$.
		\end{theorem}

        %Projective Code over F_q
        \begin{theorem}
        Let $m, r, s \in \mathbb{N}$ and $q=p^r$. Suppose $\Delta_1$ is a simplicial complex of $\mathbb{F}_q^m$ with
			\[
			\mathcal{M}=\{M_1,M_2,\dots,M_s\},
			\]
        the set of maximal elements of $\Delta_1$. 
        \begin{enumerate}
            \item Let $D = a\overline{\Delta_1^{\ast}}  \subseteq \mathcal{S}^m$ such that $\vert M_{i}\vert \geq 2$ for some $1\leq i \leq s$. Then $\tau(C_{D})$ is projective linear code with parameters
            \[
            \Big[\frac{\vert \Delta_1 \vert -1}{(q-1)},  \vert \cup_{i\in [s]} M_i \vert, q^{{\rm min}_{i\in [s]} \{\vert M_i \vert \} -1}\Big].
            \]

            \item Let $D = a\overline{\Delta_1^{c}}  \subseteq \mathcal{S}^m$ such that $\vert M_{i}\vert \geq 2$ for some $1\leq i \leq s$. Then $\tau(C_{D})$ is projective linear code with parameters
            \[
            \Big[\frac{(q^m - \vert \Delta_1 \vert)}{(q-1)},  m, q^{m-1}-\sum_{i\in [s]}q^{\vert M_i \vert -1}\Big].
            \]
        \end{enumerate}
        \end{theorem}

        \begin{corollary}\label{Proj-2wt}
            Let $r, m \in \mathbb{N}$, $q=p^r$ and let $M \subseteq [m]$. Suppose  $\Delta_M$ is a simplicial complex of $\mathbb{F}_q^m$ with the maximal element $M$.
			\begin{enumerate}
				\item \label{corol:1SF}
				Let $D = a\overline{\Delta_M^{\ast}} \subseteq\mathcal{S}^m$. Then $\tau(C_D)$ is a $1$-weight projective $q^{\vert M \vert -1}$-divisible code over $\mathbb{F}_q$ with parameters $\big[\frac{q^{\vert M\vert }-1}{q-1}, \vert M \vert, q^{\vert M \vert -1}\big]$, and hence minimal.
				
				\item \label{corol:3SF}
				Let $D = a\overline{\Delta_M^{c}} \subseteq\mathcal{S}^m$. Then $\tau(C_D)$ is a $2$-weight projective $q^{\vert M \vert -1}$-divisible code over $\mathbb{F}_q$ with parameters $\big[\frac{q^m- q^{\vert M \vert}}{q-1}, m,  (q^{m-1} - q^{\vert M \vert -1}) \big]$, where $wt_{0} = (q^{m-1} - q^{\vert M \vert -1})$ and $wt_{\infty} = q^{\vert M \vert  -1}$, and is minimal if $m > {\vert M \vert +1}$.
				
			\end{enumerate}
        \end{corollary}
    
%%%%%%%%%%%%SSS

\subsection{Secret-sharing schemes}
%Few-weight linear codes are of particular interest for their applications in secret-sharing schemes.

In this subsection, we first recall some basic concepts of secret-sharing schemes and several results that will be needed in the sequel.

Let $C$ be an $[n,k]$ linear code over $\mathbb{F}_q$ with generator matrix
\[
G=[\mathbf{g}_0,\mathbf{g}_1,\ldots,\mathbf{g}_{n-1}].
\]
Consider a secret-sharing scheme based on $C$, where the secret $s$ is an element of $\mathbb{F}_q$. The scheme consists of one dealer and $n-1$ participants, denoted by
\[
P_1,P_2,\ldots,P_{n-1}.
\]

To distribute a secret $s\in\mathbb{F}_q$, the dealer randomly selects a vector
$\mathbf{u}=(u_0,u_1,\ldots,u_{k-1})\in\mathbb{F}_q^k$ such that $s=\mathbf{u}\mathbf{g}_0.$ Since $\mathbf{g}_0\neq \mathbf{0}$, the map $\mathbf{u}\longmapsto \mathbf{u}\mathbf{g}_0$ is a nonzero linear functional on $\mathbb{F}_q^k$. Hence, for each $s\in\mathbb{F}_q$, there are exactly $q^{k-1}$ such vectors $\mathbf{u}\in\mathbb{F}_q^k$ satisfying the above equation. The dealer then computes the codeword
\[
\mathbf{t}=\mathbf{u}G=(t_0,t_1,\ldots,t_{n-1})\in C.
\]
For each $1\le i\le n-1$, the coordinate $t_i$ is assigned as the share of participant $P_i$. Note that $t_0=\mathbf{u}\mathbf{g}_0=s,$ so the first coordinate corresponds to the secret.

A subset of participants is called an \emph{access set} if its members can jointly recover the secret from their shares. The collection of all access sets is called the \emph{access structure} of the scheme. An access set is said to be \emph{minimal} if none of its proper subsets is an access set. The collection of all minimal access sets is called the \emph{minimal access structure}.

Let $\{P_{i_1},P_{i_2},\ldots,P_{i_m}\}$ be a subset of participants with corresponding shares $\{t_{i_1},t_{i_2},\ldots,t_{i_m}\}.$ Then this subset can recover the secret if and only if $\mathbf{g}_0$ belongs to the linear span of $\{\mathbf{g}_{i_1},\mathbf{g}_{i_2},\ldots,\mathbf{g}_{i_m}\}.$ If
\[
\mathbf{g}_0=\sum_{j=1}^{m}x_j\mathbf{g}_{i_j}
\]
for some $x_1,\ldots,x_m\in\mathbb{F}_q$, then
\[
s
=\mathbf{u}\mathbf{g}_0
=\sum_{j=1}^{m}x_j(\mathbf{u}\mathbf{g}_{i_j})
=\sum_{j=1}^{m}x_j t_{i_j}.
\]
Therefore, the secret can be reconstructed from the shares
$t_{i_1},t_{i_2},\ldots,t_{i_m}$.

Further details on secret-sharing schemes based on linear codes can be found in \cite{SSS,Ding_Ding}. Here we recall the following two fundamental results, from \cite{SSS-2,SSS-3}.

\begin{theorem}[\cite{SSS-2}]\label{Thm-AccessSet}
Let $C$ be an $[n,k]$ linear code over $\mathbb{F}_q$ with generator matrix
\[
G=[\mathbf{g}_0,\mathbf{g}_1,\ldots,\mathbf{g}_{n-1}].
\]
If $C$ is a minimal code, then the secret-sharing scheme based on $C^\perp$ has exactly $q^{k-1}$ minimal access sets. Moreover, if $\mathbf{g}_i$ is not a scalar multiple of $\mathbf{g}_0$ for $1\le i\le n-1$, then the participant $P_i$ belongs to exactly
$
(q-1)q^{k-2}
$
of the $q^{k-1}$ minimal access sets.
\end{theorem}

\begin{proposition}[\cite{SSS-3}]\label{Thm-CanNot}
Let $C$ be a linear code over $\mathbb{F}_q$, and let $d^\perp$ denote the minimum Hamming distance of its dual code $C^\perp$. Then, in the secret-sharing scheme based on $C$, any set of at most $d^\perp-2$ participants cannot recover the secret.
\end{proposition}

\begin{theorem}\label{ThmSSS-1}
Let $m, r, s, k \in \mathbb{N}$ and $q=p^r$. Suppose that $\Delta_1$ and $\Delta_2$ are simplicial complexes of $\mathbb{F}_q^m$ with
		\[
		\mathcal{M}=\{M_1, M_2, \dots, M_s\}
		\quad \text{and} \quad
		\mathcal{N}=\{N_1, N_2, \dots, N_k\}
		\]
		being the sets of maximal elements of $\Delta_1$ and $\Delta_2$, respectively. Let $M_i \setminus \cup_{j\in [s]\setminus \{i\}}M_j\neq \emptyset$ for any $i\in [s]$ and let $\Phi$ be the map given by Eq.~\eqref{PhiEq}.
\begin{enumerate}
%1
\item \label{part:1SSS}
Let $D=a\Delta_1+b\Delta_2\subseteq \mathcal{S}^m$ and $q^{{\rm min}_{i \in [s]}\{\vert M_i \vert\} +1} > (q-1)\sum_{ i \in [s]}q^{\vert M_i \vert}$. Then the secret-sharing scheme based on $\Phi(C_D)^\perp$ has exactly $q^{\left|\bigcup_{i\in[s]} M_i\right|-1}$ minimal access sets. Moreover, each participant belongs to precisely $(q-1)q^{\left|\bigcup_{i\in[s]} M_i\right|-2}$ of these minimal access sets. Furthermore, any collection of at most $2|\Delta_2|(q-1) q^{{\rm min}_{i\in[s]}\{|M_i|\}-1}-2$ shares cannot recover the secret.

%2
\item \label{part:2SSS}
Let $D = a\Delta_1 + b\Delta^c_2 \subseteq \mathcal{S}^m$ and $q^{{\rm min}_{i \in [s]}\{\vert M_i \vert\} +1} > (q-1)\sum_{ i \in [s]}q^{\vert M_i \vert}$. Then the secret-sharing scheme based on $\Phi(C_D)^\perp$ has exactly $q^{\vert \cup_{i\in [s]} M_i \vert-1}$ minimal access sets. Moreover, each participant belongs to precisely $(q-1)q^{\vert \cup_{i\in [s]} M_i \vert - 2}$ of these minimal access sets. Furthermore, any collection of at most $2(q^m-\vert \Delta_2 \vert )(q-1)q^{{\rm min}_{i\in [s]} \{\vert M_i \vert\} -1} - 2$ shares cannot recover the secret.

%3
\item \label{part:3SSS}
Let $D = a\Delta^c_1 + b\Delta_2 \subseteq \mathcal{S}^m$ and $q^m > \sum_{ i \in [s]} q^{\vert M_i \vert + 1}$. Then the secret-sharing scheme based on $\Phi(C_D)^\perp$ has exactly $q^{m-1}$ minimal access sets. Moreover, each participant belongs to precisely $(q-1)q^{m - 2}$ of these minimal access sets. Furthermore, any collection of at most $2\vert \Delta_2 \vert (q-1)\big(q^{m-1} - \sum_{ i \in [s]}q^{ \vert M_i \vert -1}\big) - 2$ shares cannot recover the secret.

%4
\item \label{part:4SSS}
Let $D = a\Delta^c_1 + b\Delta^c_2 \subseteq \mathcal{S}^m$ and $q^m > \sum_{ i \in [s]} q^{\vert M_i \vert + 1}$. Then the secret-sharing scheme based on $\Phi(C_D)^\perp$ has exactly $q^{m-1}$ minimal access sets. Moreover, each participant belongs to precisely $(q-1)q^{m - 2}$ of these minimal access sets. Furthermore, any collection of at most $2(q-1)(q^m - \vert \Delta_2 \vert )\big(q^{m-1} - \sum_{ i \in [s]} q^{\vert M_i \vert -1}\big) - 2$ shares cannot recover the secret.

%5
\item \label{part:5SSS}
Let $D = a\Delta_1 + b\Delta_2 \subseteq \mathcal{S}^m$ and $q^{2m} > \vert \Delta_2 \vert \sum_{ i \in [s]} q^{\vert M_i \vert + 1}$. Then the secret-sharing scheme based on $\Phi(C_{D^c})^\perp$ has exactly $q^{m-1}$ minimal access sets. Moreover, each participant belongs to precisely $(q-1)q^{m - 2}$ of these minimal access sets. Furthermore, any collection of at most $2(q-1)\big(q^{2m-1} - \vert \Delta_2 \vert \sum_{ i \in [s]} q^{\vert M_i \vert -1}\big) - 2$ shares cannot recover the secret.

\end{enumerate}
\end{theorem}

%%%Tau
\begin{theorem}\label{ThmSSS-1}
Let $m, r, s, k \in \mathbb{N}$ and $q=p^r$. Suppose that $\Delta_1$ and $\Delta_2$ are simplicial complexes of $\mathbb{F}_q^m$ with
		\[
		\mathcal{M}=\{M_1, M_2, \dots, M_s\}
		\quad \text{and} \quad
		\mathcal{N}=\{N_1, N_2, \dots, N_k\}
		\]
		being the sets of maximal elements of $\Delta_1$ and $\Delta_2$, respectively. Let $M_i \setminus \cup_{j\in [s]\setminus \{i\}}M_j\neq \emptyset$ for any $i\in [s]$ and let $\Phi$ be the map given by Eq.~\eqref{tauMap}.
\begin{enumerate}
%1
\item \label{part:1SSS}
Let $D=a\Delta_1+b\Delta_2\subseteq \mathcal{S}^m$ and $q^{{\rm min}_{i \in [s]}\{\vert M_i \vert\} +1} > (q-1)\sum_{ i \in [s]}q^{\vert M_i \vert}$. Then the secret-sharing scheme based on $\tau(C_D)^\perp$ has exactly $q^{\left|\bigcup_{i\in[s]} M_i\right|-1}$ minimal access sets. Moreover, each participant belongs to precisely $(q-1)q^{\left|\bigcup_{i\in[s]} M_i\right|-2}$ of these minimal access sets. Furthermore, any collection of at most $|\Delta_2|(q-1) q^{{\rm min}_{i\in[s]}\{|M_i|\}-1}-2$ shares cannot recover the secret.

%2
\item \label{part:2SSS}
Let $D = a\Delta_1 + b\Delta^c_2 \subseteq \mathcal{S}^m$ and $q^{{\rm min}_{i \in [s]}\{\vert M_i \vert\} +1} > (q-1)\sum_{ i \in [s]}q^{\vert M_i \vert}$. Then the secret-sharing scheme based on $\tau(C_D)^\perp$ has exactly $q^{\vert \cup_{i\in [s]} M_i \vert-1}$ minimal access sets. Moreover, each participant belongs to precisely $(q-1)q^{\vert \cup_{i\in [s]} M_i \vert - 2}$ of these minimal access sets. Furthermore, any collection of at most $(q-1)(q^m-\vert \Delta_2 \vert )q^{{\rm min}_{i\in [s]} \{\vert M_i \vert\} -1} - 2$ shares cannot recover the secret.

%3
\item \label{part:3SSS}
Let $D = a\Delta^c_1 + b\Delta_2 \subseteq \mathcal{S}^m$ and $q^m > \sum_{ i \in [s]} q^{\vert M_i \vert + 1}$. Then the secret-sharing scheme based on $\tau(C_D)^\perp$ has exactly $q^{m-1}$ minimal access sets. Moreover, each participant belongs to precisely $(q-1)q^{m - 2}$ of these minimal access sets. Furthermore, any collection of at most $\vert \Delta_2 \vert (q-1)\big(q^{m-1} - \sum_{ i \in [s]}q^{ \vert M_i \vert -1}\big) - 2$ shares cannot recover the secret.

%4
\item \label{part:4SSS}
Let $D = a\Delta^c_1 + b\Delta^c_2 \subseteq \mathcal{S}^m$ and $q^m > \sum_{ i \in [s]} q^{\vert M_i \vert + 1}$. Then the secret-sharing scheme based on $\tau(C_D)^\perp$ has exactly $q^{m-1}$ minimal access sets. Moreover, each participant belongs to precisely $(q-1)q^{m - 2}$ of these minimal access sets. Furthermore, any collection of at most $(q-1)(q^m - \vert \Delta_2 \vert )\big(q^{m-1} - \sum_{ i \in [s]} q^{\vert M_i \vert -1}\big) - 2$ shares cannot recover the secret.

%5
\item \label{part:5SSS}
Let $D = a\Delta_1 + b\Delta_2 \subseteq \mathcal{S}^m$ and $q^{2m} > \vert \Delta_2 \vert \sum_{ i \in [s]} q^{\vert M_i \vert + 1}$. Then the secret-sharing scheme based on $\tau(C_{D^c})^\perp$ has exactly $q^{m-1}$ minimal access sets. Moreover, each participant belongs to precisely $(q-1)q^{m - 2}$ of these minimal access sets. Furthermore, any collection of at most $(q-1)\big(q^{2m-1} - \vert \Delta_2 \vert \sum_{ i \in [s]} q^{\vert M_i \vert -1}\big) - 2$ shares cannot recover the secret.

\end{enumerate}
\end{theorem}

\subsection{Strongly regular graphs}
%Projective two-weight codes have been extensively investigated because of their close connections with finite projective geometries, strongly regular graphs, and combinatorial design theory. 

Strongly regular graphs (SRGs) are an important class of combinatorial structures that arise naturally in coding theory, finite geometry, and design theory. In particular, a remarkable correspondence exists between projective two-weight linear codes and strongly regular graphs. This connection not only provides a combinatorial interpretation of certain classes of linear codes but also offers an algebraic approach to generating strongly regular graphs (SRGs) with prescribed parameters. This subsection constructs several families of strongly regular Cayley graphs arising from projective two-weight codes and explicitly determine their parameters. Furthermore, we prove that the complements of these graphs are also strongly regular and determine their corresponding parameters.

Now we first recall some basic definitions and results from \cite{Few1, GodsilRoyle, delsarte1971, Few2}.

\begin{definition}[\cite{GodsilRoyle}]
	Suppose $\mathcal{V}$ is a non-empty finite set. A \textit{graph} $\mathcal{G}$ is an ordered pair $(\mathcal{V}, \mathcal{E})$, where elements of $\mathcal{V}$ are called \textit{vertices} (or \textit{nodes}), and $\mathcal{E}$ is a set of unordered pairs of distinct elements of $\mathcal{V}$, called \textit{edges}.
\end{definition}

\begin{definition}[Cayley Graph]
	Suppose $(\mathcal{H}, +)$ is a finite additive group and $S$ is a subset of $\mathcal{H}$ satisfying $\textbf{0}_\mathcal{H} \notin S$ and $S = S^{-1}$ (that is, $s \in S$ implies $-s \in S$).
	The \emph{Cayley graph} $\mathcal{H}(S)$ is the undirected graph whose vertex set is $\mathcal{V} = \mathcal{H}$, and any two distinct vertices $x, y \in \mathcal{H}$ are adjacent if and only if $x-y \in S$ , that is, the edge set is
		\[
		\mathcal{E} = \bigl\{\, \{\mathbf{x}, \mathbf{y}\}  \mid \textbf{x}, \textbf{y} \in \mathcal{V},
		\mathbf{x} - \mathbf{y} \in \Omega \,\bigr\}.
		\]
\end{definition}

\begin{definition}[Complement Graph]
Let $\mathcal{G}=(\mathcal{V},\mathcal{E})$ be a graph. The \emph{complement} of $\mathcal{G}$, denoted by $\overline{\mathcal{G}}$, is the graph with vertex set $\mathcal{V}$ and edge set
\[
\overline{\mathcal{E}}
=
\bigl\{
\{u,v\}
\,:\,
u,v\in\mathcal{V},\ u\neq v,\ \{u,v\}\notin\mathcal{E}
\bigr\}.
\]
\end{definition}

\begin{definition}[Strongly Regular Graph]
A graph $\mathcal{G}=(\mathcal{V},\mathcal{E})$ is said to be a \emph{strongly regular graph} with parameters $(N,K,\lambda,\mu)$ if it has $N$ vertices, is $K$-regular, and satisfies the following properties:
\begin{itemize}
    \item every pair of adjacent vertices has exactly $\lambda$ common neighbors;
    \item every pair of nonadjacent vertices has exactly $\mu$ common neighbors.
\end{itemize}
\end{definition}

Let $C$ be an $[n,k]$ linear code over $\mathbb{F}_q$ with generator matrix $G=[\mathbf{x}_1,\mathbf{x}_2,\ldots,\mathbf{x}_n],$
where $\mathbf{x}_i\in\mathbb{F}_q^k$ is the $i^{th}$ column of $G$. Assume that $C$ is projective, i.e., no two columns of $G$ are scalar multiples of one another.

Let $V=\mathbb{F}_q^k$ and define $\mathcal{O}
=
\{\langle \mathbf{x}_i\rangle : 1\le i\le n\},$
where $\langle \mathbf{x}_i\rangle$ denotes the $1$-dimensional subspace of $V$ generated by $\mathbf{x}_i$. Thus, $\mathcal{O}$ may be viewed as a set of points in the projective space $\mathrm{PG}(k-1,q)$.

Define $\Omega =
\{\mathbf{v}\in V\setminus\{\mathbf{0}\} :
\langle \mathbf{v}\rangle\in\mathcal{O}\}.$ In other words, $\Omega =
\bigcup_{i=1}^{n}
\{\alpha\mathbf{x}_i:\alpha\in\mathbb{F}_q^\ast\}.$

Associated with $C$, we define the graph $\mathcal{G}(\Omega)$ whose vertex set is $V=\mathbb{F}_q^k$, and in which two distinct vertices $\mathbf{x},\mathbf{y}\in V$ are adjacent if and only if $\mathbf{x}-\mathbf{y}\in\Omega.$ Hence, $\mathcal{G}(\Omega)$ is the Cayley graph $\mathrm{Cay}(V,\Omega)$ on the additive group $(\mathbb{F}_q^k,+)$ with connection set $\Omega$.

In \cite{Few1}, the authors studied a connection between projective two-weight linear codes and strongly regular graphs. We recall this result below.

\begin{theorem}[\cite{Few1}]\label{SRG-main}
Let $C$ be a projective linear code over $\mathbb{F}_q$, and let $\mathcal{G}(\Omega)$ be the Cayley graph associated with $C$. Then $\mathcal{G}(\Omega)$ is a strongly regular graph if and only if $C$ is a projective two-weight code.
\end{theorem}

\begin{theorem}[\cite{Few1}]\label{SRG-Param}
Let $C$ be an $[n,k]$-projective two-weight linear code over $\mathbb{F}_q$ with nonzero weights $w_1$ and $w_2$. Then the strongly regular graph $\mathcal{G}(\Omega)$ associated with $C$ has parameters $(N,K,\lambda,\mu)$, where $N=q^k,\ K=n(q-1),$  $\lambda = K^2+3K-q(w_1+w_2)-Kq(w_1+w_2)+q^2w_1w_2,$ and $\mu = \frac{q^2w_1w_2}{q^k}.$
\end{theorem}

\begin{theorem}[\cite{Ref-SRG}]\label{Thm-SRG-Complement}
Let $\mathcal{G}$ be a graph. Then $\mathcal{G}$ is strongly regular if and only if its complement $\overline{\mathcal{G}}$ is strongly regular.

Moreover, if $\mathcal{G}$ is a strongly regular graph with parameters $(N,K,\lambda,\mu),$ then its complement $\overline{\mathcal{G}}$ is a strongly regular graph with parameters $\bigl(N,\,
N-K-1,\,
N-2K+\mu-2,\,
N-2K+\lambda
\bigr).$
\end{theorem}
By applying Corollary~\ref{Proj-2wt}, we obtain the following result.
\begin{theorem}\label{Thm-SRG}
	Let $r, m \in \mathbb{N}$, $q=p^r$ and let $M \subset [m]$ such that $\vert M \vert < m \leq 2\vert M \vert$.
	Let $D = a\overline{\Delta_M^{c}} \subseteq\mathcal{S}^m$. Then the Cayley graph $\mathcal{G}(\Omega)$ corresponding to $\tau(C_D)$ is strongly regular with parameters
	\begin{equation*}
		\left(q^m,\, q^m-q^{|M|},\, 2q^m-3q^{|M|},\, q^{|M|}-q^{2|M|-m} \right).
	\end{equation*}
	Furthermore, if $\vert M \vert = m-1$, then complement graph $\overline{\mathcal{G}(\Omega)}$ is also strongly regular with parameters
	\begin{equation*}
		\left( q^m,\, q^{|M|}-1,\, 3q^{|M|}-q^m-q^{2|M|-m}-2,\, q^m-q^{|M|} \right).
	\end{equation*}
\end{theorem}

\section{Conclusions}\label{Sec:8}
        In this paper, we studied a class of $\mathcal{S}$-linear codes constructed from simplicial complexes, where $\mathcal{S}$ is a finite extension of a commutative non-unital ring $I$ of order $p^2$. The corresponding Gray image codes and subfield-like codes are also investigated and obtained several classes of few-weight divisible codes over $\mathbb{F}_q$. We established conditions under which these codes are minimal, optimal and self-orthogonal. As applications, we constructed families of locally recoverable codes (LRCs) and determined their locality. We also presented two infinite classes of projective few-weight codes. Furthermore, we investigated the access structures of secret-sharing schemes based on the duals codes. Finally, we obtained a class of strongly regular graphs and explicitly determined their parameters. Moreover, we showed that the complement of these graphs are strongly regular and computed their corresponding parameters.

        %These results demonstrate that simplicial complexes over finite ring extensions provide a rich framework for constructing linear codes with desirable combinatorial and algebraic properties.

		\section*{Acknowledgment}
		This work was supported by the ANRF National Postdoctoral Fellowship (NPDF) under Grant No.~PDF/2025/004899.

	\end{document}